\documentclass[letterpaper,accepted=2022-02-11,allowfontchangeintitle]{quantumarticle}
\pdfoutput=1
\usepackage[ascii]{inputenc}
\usepackage[english]{babel}
\usepackage[T1]{fontenc}
\usepackage[numbers,sort&compress]{natbib}
\usepackage[colorlinks=true,hyperindex,hypertexnames=false,breaklinks,pdfusetitle,bookmarksnumbered]{hyperref}
\usepackage{amsmath,amssymb,bbold,amsthm,graphicx,dsfont,xcolor,appendix,algorithm,algpseudocode,mathtools}
\usepackage[capitalize,nameinlink]{cleveref}
\DeclareMathOperator{\tr}{Tr}
\DeclareMathOperator{\iSWAP}{iSWAP}
\DeclareMathOperator{\XX}{XX}
\DeclareMathOperator{\XY}{XY}
\DeclareMathOperator{\YY}{YY}
\DeclareMathOperator{\Pf}{Pf}
\DeclareMathOperator{\SO}{SO}
\DeclareMathOperator{\OR}{O}
\DeclareMathOperator{\Spin}{Spin}
\DeclareMathOperator{\Pin}{Pin}

\newcommand{\so}{\mathfrak{so}}
\newcommand{\id}{\mathbb{1}}
\newcommand{\dens}[1]{\lvert#1\rangle\!\langle#1\rvert}
\newcommand{\ket}[1]{\lvert#1\rangle}
\newcommand{\bra}[1]{\langle#1\rvert}
\newcommand{\mbb}[1]{\mathbb{#1}}

\newcommand{\mc}[1]{\mathcal{#1}}
\newcommand{\md}[1]{\mathds{#1}}

\newcommand{\ct}{^\dagger}
\newcommand{\tn}[1]{^{\otimes #1}}

\newcommand{\vsp}{\vphantom{{{\sum^i}^I}^I}}
\newcommand{\ot}{\otimes}
\newtheorem{theorem}{Theorem}
\newtheorem*{theoremone}{Theorem~\ref{thm:var_bound}}

\newtheorem{lemma}[theorem]{Lemma}

\begin{document}
\title{\LARGE Matchgate benchmarking: Scalable benchmarking of a continuous family of many-qubit gates}
\author{Jonas Helsen}
\affiliation{QuSoft \& Korteweg-de Vries Institute for Mathematics, University of Amsterdam, Science Park 123, 1098 XG Amsterdam, The Netherlands}
\affiliation{Centrum Wiskunde \& Informatica (CWI), Science Park 123, 1098 XG Amsterdam, The Netherlands}
\author{Sepehr Nezami}
\affiliation{Institute for Quantum Information and Matter, Caltech, Pasadena, CA 91125, USA}
\author{Matthew Reagor}
\affiliation{Rigetti Computing, 775 Heinz Ave, Berkeley, CA 94710, USA}
\author{Michael Walter}
\affiliation{QuSoft \& Korteweg-de Vries Institute for Mathematics, University of Amsterdam, Science Park 123, 1098 XG Amsterdam, The Netherlands}
\affiliation{Institute for Theoretical Physics \& ILLC, University of Amsterdam, Science Park 123, 1098 XG Amsterdam, The Netherlands}
\affiliation{Faculty of Computer Science, Ruhr University Bochum, Universit\"atsstra\ss{}e 150, 44801 Bochum, Germany}
\hypersetup{pdfauthor={Jonas Helsen, Sepehr Nezami, Matt Reagor, Michael Walter},pdftitle = {Matchgate benchmarking: Scalable benchmarking of a continuous family of many-qubit gates}}
\begin{abstract}
We propose a method to reliably and efficiently extract the fidelity of many-qubit quantum circuits composed of continuously parametrized two-qubit gates called matchgates.
This method, which we call \emph{matchgate benchmarking}, relies on advanced techniques from randomized benchmarking as well as insights from the representation theory of matchgate circuits.
We argue the formal correctness and scalability of the protocol, and moreover deploy it to estimate the performance of matchgate circuits generated by two-qubit XY spin interactions on a quantum processor.
% \JH{499/600 characters}
\end{abstract}
\maketitle

Quantum computers promise a revolution in computational power, and a multinational effort is underway to construct them.
One of the key challenges in the building and operating of quantum computers is the appearance of errors in computations, either due to inaccuracies in control or due to interactions with the environment.
% These errors can be overcome using techniques from quantum error correction, but only if the error rate is low enough.
It is thus vitally important to be able to characterize accurately and efficiently the type and magnitude of errors present in quantum operations.
To this end, a variety of techniques have been developed, with the most popular class of techniques known as randomized benchmarking (RB)~\cite{hashagen2018real,helsen2019new,cross2016scalable,carignan2015characterizing,wallman2015robust,PhysRevA.90.030303,GambettaEtAl:2012:simultaneousRB,knill2008randomized,francca2021efficient}, where one characterizes the quality of gates in a gateset by applying random sequences of gates of increasing length, and tracks the corresponding increase in average error.
For a recent overview of RB protocols see~\cite{framework} and references therein.

Randomized benchmarking has been extremely successful in characterizing quantum operations on a variety of platforms~\cite{xue2019benchmarking,PhysRevA.90.030303,CycleBenchmarking}.
Yet it suffers from a number of shortcomings that limit its usefulness in some important situations.
Firstly, standard RB protocols mix the error contributions of various types of gates (such as single qubit gates and two qubit gates) and only report an average error.
This is problematic because different types of gates are created by different physical mechanisms and hence have different error contributions.
Moreover, different types of gates contribute differently to the error thresholds that must be met for fault-tolerance.
For instance it is often the case that more stringent requirements are imposed on two-qubit gates than on single-qubit gates.
Secondly, standard randomized benchmarking protocols only test discrete gatesets (such as the Clifford group), while continuously parametrized gatesets are vital for near-term quantum computing applications such as VQE and QAOA~\cite{mcclean2016theory,farhi2014quantum}.
For these reasons it is desirable to devise gate assessment procedures that combine the proven advantages of randomized benchmarking with the ability to handle continuous gate families and focus on a single type of quantum gates.

In this paper, we address this challenge by proposing \emph{matchgate benchmarking}, an advanced randomized benchmarking procedure based on the general framework given in~\cite{framework} as well as the recently introduced linear cross-entropy benchmarking~\cite{arute2019quantum}.
Our procedure natively uses continuously parametrized two-qubit gates and estimates fidelities in a scalable way, both in terms of statistical sampling and classical computational resources.
To prove its correctness and scalability we use techniques from the representation theory of the matchgate group.
Moreover we provide an implementation of the protocol on a small quantum computer, showing that our protocol can reliably assess the quality of two qubit gates in a realistic environment.

Because this paper presents both proposals for experimental procedures and mathematical results, and thus is aimed at both experimental practitioners and theorists, we defer all technical proofs to appendices.

%=============================================================================
\section{Matchgates}
%=============================================================================
Matchgate circuits are a continuous class of quantum circuits originally conceived by Valiant~\cite{valiant2002expressiveness} and subsequently connected to the theory of free fermions by Knill~\cite{knill2001fermionic} and Terhal-DiVincenzo~\cite{terhal2002classical}, see also~\cite{divincenzo2005fermionic,bravyi2005lagrangian,jozsa2008matchgates}.
They are explicitly realized by~$\XY$ or by~$\XX$ and~$\YY$ spin interactions and are thus the natural choice for two-qubit gates on many physical platforms such as ion traps.
The standard~$\iSWAP$~\cite{schuch2003natural} and~$\XY(\theta)$ gates~\cite{abrams2019implementation} are examples of matchgates (though they do not generate the full group by themselves).
A key property of matchgate circuits is that they are efficiently classically simulable (like the better known Clifford group often used in standard RB), which is a necessary requirement for a scalable randomized benchmarking procedure.
The two-qubit matchgates are generated by unitaries of the form~$U(\alpha) = \exp(i \alpha P\otimes P')$, where~$P,P'$ are Pauli $X$ or~$Y$ matrices.
The $n$-qubit \emph{matchgate group}~$\mathcal{M}_n$ is then defined by considering $n$ qubits on a line and composing \emph{nearest-neighbor} gates of this form, along with single-qubit $Z(\theta) = \exp(i \theta Z)$ gates. We further extend this group with a Pauli $X$ on the \emph{last} qubit. This forms the $\mathcal{M}^+_n = \langle \mc{M}_n, X_n\rangle$ group, which we will refer to as \emph{generalized matchgates}.

Matchgates are intimately connected to non-in\-ter\-act\-ing fermions. This connection is key to their efficient simulation.
To see this, consider the Majorana fermion operators (represented on qubits by the Jordan-Wigner isomorphism)
\begin{align*}
  \gamma_{2j-1} &= Z_1 \ldots Z_{j-1}X_j I_{j+1}\ldots I_n, \\
  \gamma_{2j} &= Z_1 \ldots Z_{j-1}Y_j I_{j+1}\ldots I_n,
\end{align*}
with $j \in [n] :=\{1,\dots,n\}$ and $I,X,Y,Z$ the Pauli matrices, the subscript indicating on which qubit they act.
Any matchgate unitary $U\in \mc{M}_n$ acts by conjugation on the Majorana operators as
\begin{equation}\label{eq:match_to_orth}
  U\gamma_jU\ct = \sum_{j\in [2n]}R_{kj}\gamma_k,
\end{equation}
with $R\in \SO(2n)$ a rotation matrix. Moreover the $X_n$ gate maps all $\gamma_i$ for $i<2n$ to themselves but maps~$\gamma_{2n}$ to~$-\gamma_{2n}$, so it corresponds to a reflection $F$ of the $2n$-th axis.
In this way, matchgate and generalized matchgates unitaries can be efficiently tracked on a classical computer.
Moreover, for any rotation there is a matchgate unitary $U=U(R)$ implementing it, and for every element $Q$ of $\OR(2n)$ (rotations plus reflections) there is a corresponding generalized matchgate unitary $U=U(Q)$. In fact, $\mc{M}_n$ is generated as a Lie group by Hamiltonians of the form $H = \frac i4 \sum_{j,k\in [2n], j\neq k} \alpha_{jk} \gamma_j\gamma_k$, with $\alpha$ a real antisymmetric $2n\times 2n$ matrix, so $\mc M_n$ can be understood as a representation of $\Spin(2n)$ and $\mc{M}_n^+$ a representation of $\Pin(2n)$.
Since the Majorana operators and their products span the space of $n$-qubit operators, (generalized) matchgate unitaries are fully determined by the corresponding rotation (and reflection) matrix (up to an overall phase).

Before we can define our benchmarking procedure we need to briefly discuss the action of the generalized matchgate group on the space of $n$-qubit operators.
We denote products of Majoranas as $\gamma[S] =\prod_{s \in S} \gamma_s$ for $S\subseteq[2n]$, with $\gamma_\emptyset = I$ and the product taken in increasing order.
For each $k\in[2n]$, consider the subspace $\Gamma_k = \langle \gamma[S] \;|\; S\subseteq [2n],\; \lvert S\rvert=k\rangle$ spanned by products of $k$ Majorana operators.
Then, for $k\in [2n]$ each $\Gamma_k$ is an irreducible representation of the generalized matchgate group $\mc{M}_n^+$. Moreover all these representations are inequivalent. A proof of this statement is given in \cref{lem:rep_th}.

We note that the addition of the extra bit flip gate, lifting the matchgates to the generalized matchgates, is critical in ensuring the mutual inequivalence of the representations $\Gamma_k$.

%=============================================================================
\section{Matchgate benchmarking}
%=============================================================================
The matchgate benchmarking protocol, given formally in \cref{prot:protocol}, estimates the quality of generic circuits in the generalized matchgate group $\mc{M}_n^+$ in a manner that scales efficiently with the number of qubits and is resistant to state preparation and measurement (SPAM) errors.
The output of matchgate benchmarking is a list of decay parameters $\lambda_{k}$ that characterize the noise associated to the subspace~$\Gamma_k$.
We will call these decay parameters \emph{Majorana fidelities}.

The protocol consists of multiple rounds with varying parameters~$k \in [2n]$ and sequence lengths~$m$.
Each round starts with the preparation of either the all-zero $\ket{0_n}:= \ket{0}\tn{n}$ or the all-plus $\ket{+_n}:=\ket{+}\tn{n}$ state.
This is followed by~$m$ generalized matchgate unitaries $U(Q_1),\dots,U(Q_m)$, chosen uniformly and independently at random from $\mc{M}_n^{+}$ (we describe an efficient method for sampling below).
Finally, all qubits are measured in either the computational ($Z$) basis or the Hadamard ($X$) basis.
We write~$\rho_0$ for the initial state and~$\{E_x\}_{x\in\{0,1\}^n}$ for the measurement POVM, and refer to the two SPAM settings as $X$ and $Z$-basis SPAM.
The above is repeated many times until the relative frequencies~$f_x$ of the measurement outcomes $x\in\{0,1\}^n$ give a good estimate of the true probabilities, which we denote by~$p(x|Q,m)$, where $Q = Q_m\cdots Q_1$.
Since $Q$ itself is uniformly random, by averaging over many random sequences we thus obtain a good estimate $\hat f_k(m)$ of the weighted average
\begin{align}\label{eq:av_data}
  f_k(m) = \int_{\OR(2n)} dQ \sum_{x\in \{0,1\}^n} \, \alpha_k(x,Q) p(x|Q,m),
\end{align}
where we use the \emph{correlation function} $\alpha_k$ defined by
\begin{align}\label{eq:corr_func}
  \alpha_k(x,Q) = \frac1{N_k} \! \tr\Bigl(E_x P_k\bigl( U(Q)\rho_0 U(Q)\ct\bigr) \Bigr).
\end{align}
Here $P_k$ denotes the projection superoperator onto the subspace~$\Gamma_k$, and the normalization constant~$N_k = 2^{-n} \binom{n}{\lfloor k/2\rfloor}^2 \binom{2n}{k}^{-1}$ is chosen so that~$f_k(m) = 1$ if the gates are perfectly implemented.
The correlation functions~$\alpha_k$ can be efficiently computed using the classical simulation techniques for matchgates~\cite{terhal2002classical,knill2001fermionic} as well as several tricks for evaluating Pfaffian sums~\cite{ishikawa2006applications}.
We give explicit expressions in \cref{appsec:corr_func}.

%=============================================================================
\section{Interpretation and analysis}
%=============================================================================
Intuitively the randomization over gates $U(Q_1), \ldots, U(Q_m)$ averages out the noise associated to each gate, resulting in a linear combination of generalized depolarizing channels, one for each irreducible subspace~$\Gamma_k$.
This observation forms the basis of the randomized benchmarking approach, of which matchgate benchmarking is an example~\footnote{For comparison: In standard RB with the Clifford group there are two subrepresentations, giving rise to a standard depolarizing channel upon averaging.}.
In our setting, the associated depolarization parameters are described by what we call the \emph{Majorana fidelities} $\lambda_k$, since the corresponding subspaces $\Gamma_k$ are spanned precisely by the $k$-fold product of the Majorana operators.
The role of the correlation functions~$\alpha_k$ is precisely to address the individual subspaces $\Gamma_k$ and thus isolate the corresponding Majorana fidelities.
In particular $f_k(m)$ will be determined precisely by the $m$-th power $\lambda_k^m$.
To make this more concrete, suppose that each generalized matchgate unitary $U(Q)$ is realized by a quantum channel $\Phi(Q)$ that describes its actual implementation.%
\footnote{The existence of such a map~$\Phi$ is an assumption on the underlying device and in particular excludes time- and context-dependent effects. It is, however, the weakest assumption under which RB protocols can be guaranteed to function correctly~\cite{framework}.}
If we assume \emph{gate-independent noise}, i.e., that there is a quantum channel $\Lambda$ such that $\Phi(Q)(\rho) = \Lambda(U(Q) \rho U(Q)\ct)$ for all $Q \in \OR(2n)$, then the average~$f_k(m)$ is described \emph{exactly} by a single exponential decay
\begin{equation*}%\label{eq:mexp decay}
  f_k(m) = A_k \lambda_k^m.
\end{equation*}
Importantly, $\lambda_k$ depends only on the noise channel~$\Lambda$, while $A_k$ also depends on the SPAM. A proof of of the above statement is given in \cref{appsec:gate_indep}
The assumption of gate-independent noise is unrealistic, but it can be relaxed significantly using Fourier analytic techniques~\cite{framework}. We outline an approach to this in \cref{appsec:gate_dep} but leave a detailed derivation for future work.

In the absence of noise, the Majorana fidelities are equal to one, so their deviation from the identity encodes properties of the noise.
In particular, the average gate fidelity can (for gate-independent noise) be recovered by
\begin{equation}\label{eq:fid}
% 2^{-n} \sum_{k=0}^n \dim \Gamma_{k} \tr(M_k) =
2^{-n} \sum_{k=0}^{2n} \!\binom{2n}{k} \lambda_{k}
= \! (2^n \!+\!1)F_{\mathrm{avg}}(\Lambda) - 1.
\end{equation}
% Moreover, for specific values of $k$ the Majorana fidelities %eigenvalues of the matrices~$M_k$
% have an independent operational interpretation.
% For instance, $\lambda_{0}$ and $\lambda_{2n}$ encode properties of leakage and parity violation of the noise respectively. This can be seen by noting that the representations $\Gamma_0$ and $\Gamma_{2n}$ are spanned by the identity matrix and the parity operator $Z_1\ldots Z_n$ respectively. Thus if the noise channel is trace and parity preserving, both will be equal to one.
Interpreting the individual values $\lambda_k$ operationally is less straightforward.
The parameter $\lambda_0$ has a well-known interpretation as a measure of trace-loss of the channel~$\Lambda$.
Moreover, if $\Lambda$ is unitary (i.e.
generated by some Hamiltonian $H$) then the parameter $\lambda_{2n}$ can be seen as a measure of parity preservation of this evolution.
More generally we can interpret the parameters $\lambda_k$ if we restrict the channel $\Lambda$ to be Gaussian (as discussed in \cite{bravyi2005lagrangian}) and unital.
In this case the channel $\Lambda$ has associated to it a $2n\times 2n$ real matrix $B$ s.t.\ $BB^T\leq \id$ and the action on a Majorana operator $\gamma[S]$ is defined as $\Lambda(\gamma[S]) = \sum_{S'\subset [2n]\,,|S'|=|S|} \det(B[S',S])\gamma[S']$, where $B[S',S]$ denotes the submatrix of $B$ with rows in $S'$ and columns in $S$.
Hence the decay parameter $\lambda_k$ is in this case precisely given by the average over the principal minors of size $k$.
It is easy to see that these averages over principal minors precisely correspond to (normalized) elementary symmetric polynomials of the eigenvalues of $B$.
Hence the Majorana fidelities encode eigenvalue information of Gaussian noise.
These eigenvalues can then be in principle extracted by solving the system of $2n$ polynomials in $2n$ unknowns (although one would have to work out the stability of solutions under statistical noise, which we do not attempt here).

Given the functional form of $f_k(m)$, we can extract the value of $\lambda_k$ by fitting to a single  exponential decay.
To perform this fit in practice, a correct choice of SPAM operators $\rho_0, \{E_x\}$ is required that ensures that the pre-factors $A_k$ are large enough.
In the noise-free limit the parameters $A_k$ can be explicitly computed (which we do in \cref{appsec:gate_indep}).
For odd~$k$ (and $X$-basis SPAM), we find that $A_k =1$, and thus the fitting problem is well conditioned.
For even~$k$ (and $Z$-basis SPAM), we similarly find that $A_k=1$.
% It can thus happen in principle that $M_k$ has eigenvalues that are not visible in the data $f_k(m)$.
% However, given reasonably high quality gates, the matrix $M_k$ will be close to the identity and thus this issue is not relevant in practice.
% and $A_k = \begin{psmallmatrix}1 & 0 \\ 0 & 0\end{psmallmatrix}$ which is a rank one matrix. Moreover we will see than $M_n$ is diagonal in the same basis so the $\lambda_{n,1}$ eigenvalue wil not be detected
We also note that if we choose $Z$-basis SPAM and consider odd $k$, we have $A_k=0$, giving no visibility. Similarly we have $A_{2n}=0$ for $X$-basis SPAM.
This motivates our use of different SPAM settings for different values of $k$.

\begin{figure}
\begin{algorithm}[H]
\caption{Matchgate benchmarking}\label{prot:protocol}
\algrenewcommand\algorithmicindent{0em}
\begin{algorithmic}[1]
\For{$k \in \{0,\dots,2n\}$}
\For{$m \in \text{sequence lengths}$} \vskip 0.2em
\algrenewcommand\algorithmicindent{1em}
\For{$i \in [K]$}
  \State Prepare $\ket{0}\tn{n}$ ($k$ even) or $\ket{+}\tn{n}$ ($k$ odd).
  \For{$j \in [m]$}
  \State Apply $Q_j^{(i)}\in \OR(2n)$ chosen uniformly at random.
  \EndFor.
  \State Measure in the $Z$ ($k$ even) or $X$ basis ($k$ odd).
  \State Repeat the above many times and record frequencies $f^{(i)}_x$ of measurement outcomes $x\in\{0,1\}^n$.
  % \State Compute correlation functions $\alpha^{(i)}_x = \alpha_k(x,R_m^{(i)} \cdots R_1^{(i)})$.
\EndFor \vskip 0.2em
\State Compute the empirical average \vskip -1.5em
\begin{align}\label{eq:av_emp}
  \!\!\!\!\hat{f}_k(m) = \frac{1}{K}\sum_{i=1}^K \sum_{x\in \{0,1\}^n} \alpha_k(x, Q_m^{(i)} \cdots Q_1^{(i)}) f^{(i)}_x.
\end{align}
\EndFor
\State Fit $\{\hat{f}_k(m)\}_m$ to $\hat{f}_k(m) =_{\mathrm{fit}} A_{k}\lambda^m_{k}$.
\EndFor
\State Output the Majorana fidelities $\{\lambda_{k}\}_{k=0}^{2n}$.
\end{algorithmic}
\end{algorithm}
\end{figure}

%=============================================================================
\section{Generating random matchgate circuits}
%=============================================================================
Random rotations in $\OR(2n)$ (and thus generalized matchgate unitaries) can readily be sampled in an efficient manner by a variety of methods.
However it is desirable to generate them directly in the form of circuits involving one- and two-qubit matchgates. Begin by noting that we can decompose any generalized matchgate $U(Q)$ as $ U(Q) = U(R)X_n^b$ with $b\in \{0,1\}$ where $U(R)$ (with $R\in \SO(2n)$) is a matchgate and $X_n$ is a bit-flip on the last qubit. Hence the task of sampling random generalized matchgates reduces to that of sampling random matchgates. For this we give a method based on the probabilistic Hurwitz lemma~\cite{diaconis2000bounds}.
Consider the rotation
\begin{align}\label{eq:hurwitz}
  R =
  (G^{(1)}_{2n\!-\!1}\!\cdots\! G^{(1)}_1)
  (G^{(2)}_{2n\!-\!1}\!\cdots\! G^{(2)}_2)
  \!\cdots\!
  % (G^{(2n\!-\!2)}_{2n\!-\!1} G^{(2n\!-\!2)}_{2n\!-\!2})
  G^{(2n\!-\!1)}_{2n\!-\!1},
\end{align}
where each $G^{(i)}_j$ is a two-dimensional rotation by a random angle in the $j,j+1$-plane.
Proposition~2.1 in~\cite{diaconis2000bounds} implies that $R$ is a uniformly (Haar) random rotation in $\SO(2n)$.
The formula for $R$ translates directly into a circuit for the corresponding matchgate unitary $U(R)$ that only involves single qubit $Z_j$ and two-qubit $X_j X_{j+1}$ rotations.
Since $Z$-rotations are virtual~\cite{mckay2017efficient,xue2019benchmarking}, and hence noiseless, in many platforms, the dominant source of noise in this circuit is from the two qubit gates. Adding the aforementioned random $X_n$ gates we obtain a uniformly random circuit $U(Q)\in \mc{M}_n^{+}$ generalized matchgate circuit $U(Q)$ given above is comparable to that of a generic $n$-qubit Clifford gate~\cite{koenig2014efficiently}.

Finally, if one only has access to $\XY$ gates (corresponding to $\XX+\YY$ rotations, which do not by themselves generate the full matchgate group) as opposed to $\XX$ or $\YY$ rotations, the above construction can be implemented by the identity $\XX(\theta) = \XY(\theta/2) \, X_1 \, \XY(\theta/2) \, X_1$, where~$X_1$ denotes the bit-flip gate on the first qubit.

We conjecture that one can also efficiently sample \emph{approximately uniform} matchgate unitaries by repeatedly choosing nearest-neighbor pairs of qubits at random and applying a random element of $\mc{M}_2$.
This is a variation of the well-known Kac random walk on~$\SO(2n)$ which mixes to approximate uniformity in polynomial time~\cite{jiang2017kac}.
This maybe possibly be an even more efficient way of sampling (approximately) uniformly random matchgate circuits.

%=============================================================================
\section{Statistical scalability}
%=============================================================================
We now consider the scalability of the matchgate benchmarking protocol with respect to the number of qubits.
Recall that in \cref{prot:protocol} we determine the relative frequencies of measurement outcomes for a number of random matchgate sequences.
By taking an empirical average, one obtains an estimate~$\hat{f}_k(m)$ for~$f_k(m)$.
It is a priori unclear whether the variance of this estimate might grow exponentially with the number of qubits $n$, rendering the estimation procedure infeasible beyond a few qubits.
We argue this is not the case by explicitly bounding the variance in the noise-free limit.

\begin{theorem}\label{thm:var_bound}
Consider the estimator $\hat{f}_k(m)$ for the quantity $f_k(m)$ defined in \cref{eq:av_data,eq:av_emp}.
Assuming no noise, its variance is bounded (uniformly in $k$ and $m$) as
\[
\md{V}(\hat{f}_k(m)) = \frac{1}{K} O(\mathrm{poly}(n)).
\]
\end{theorem}
A proof of this theorem is given in \cref{appsec:variance}.
The central ingredient in this theorem is a novel moment bound for random matchgates based on the representation theory of $\SO(2n)$, which may be of independent interest:
\begin{lemma}\label{lem:moment_bound}
Let $\ket{\theta}$ be the all-zero ($\ket{0}\tn{n}$) or the all-plus ($\ket{+}\tn{n}$) state, and let $t$ be a fixed integer. Then:
\begin{equation*}
\int_{\SO(2n)} |\bra{\theta} U(Q) \ket{\theta}|^{2t} = 2^{-tn} \, O(\mathrm{poly}(n)).
\end{equation*}
\end{lemma}
% This is reminiscent of the Weingarten integral formula for the unitary group.
A proof of this lemma is given in \cref{appsec:moments}.
We note that our variance upper bound in \cref{thm:var_bound} is likely loose and we expect the real variance to be substantially smaller.
The theorem can also be extended to the case of gate-dependent noise at the cost of some technical complications, but we do not pursue this here.

%=============================================================================
\section{Experimental demonstration}
%=============================================================================

We apply the matchgate benchmarking protocol (\cref{prot:protocol}) to benchmark the native $XY(\theta)$ gate~\cite{abrams2019implementation} between two qubits on the Rigetti Aspen-8 chip. A collection of circuits was run consecutively with fixed sequence lengths ranging from $m=2$ to $m=24$, sampling $K=64$ random sequences of orthogonal matchgates per sequence length, and performing $L=400$ measurement repetitions (shots) per sequence; for a total of $5.9\times 10^5$ individual shots.
The code and data for this experiment can be found at~\cite{data}.
\Cref{fig:exp-data} shows the results of this experiment. All error bars are bootstrapped $95\%$ confidence intervals.
For $n=2$, there are five exponential decays, associated to $k\in [4]$.
From fitting the experimental data to single exponentials we obtain the following values for the Majorana fidelities:
% \begin{center}
% {\small \begin{tabular}{|c|c|c|c|c|}
% \hline
% $\lambda_0$& $\lambda_1$&$\lambda_2$&$\lambda_3$&$\lambda_4$\\\hline
% ${1.000\!\pm\!0.001}$&${0.78\!\pm \!0.05}$&${0.85\!\pm \!0.02}$&${0.87\!\pm\! 0.02}$&${0.83\!\pm \!0.02}$\\
% \hline
% \end{tabular}}
% \end{center}
\begin{center}
\small \begin{tabular}{|c|c|}
\hline
$\lambda_0$& ${1.000\!\pm\!0.001}$\\
\hline
$\lambda_1$& ${0.78\!\pm \!0.05}$\\
\hline
$\lambda_2$& ${0.85\!\pm \!0.02}$\\
\hline
$\lambda_3$& ${0.87\!\pm\! 0.02}$\\
\hline
$\lambda_4$& ${0.83\!\pm \!0.02}$\\
\hline
\end{tabular}
\end{center}
% \begin{center}
% \begin{tabular}{|c|c|}
% \hline
% Majorana fidelity & estimated value\\\hline
% $\lambda_0$ & ${1.000\pm0.001}$\\
% $\lambda_1$ & ${0.78\pm 0.05}$\\
% $\lambda_2$ & ${0.85\pm 0.02}$\\
% $\lambda_3$ & ${0.87\pm 0.02}$\\
% $\lambda_4$ & ${0.83\pm 0.02}$\\
% \hline
% \end{tabular}
% \end{center}
%
From these decays and \cref{eq:fid} we can infer that the expected average fidelity of a random two-qubit matchgate is~$F={0.88\pm0.02}$.
Since a random two qubit matchgate requires four~$XY(\theta)$ gates, we can make a heuristic lower bound estimate of the average fidelity of the $XY(\theta)$ gate by assuming that the single-qubit gates are essentially noiseless and that the fidelity is approximately multiplicative, finding~$F={0.97}$ for a single $XY(\theta)$ gate. We compare this estimate to standard interleaved RB applied to the same~$\iSWAP$ gate ($XY(\pi)$), where we observe~$F={0.986\pm0.006}$, consistent with the fidelity range for $XY$-gates observed in~\cite{abrams2019implementation}. We attribute additional error in the matchgate construction to single-qubit rotations.

\begin{figure}
\includegraphics[scale=0.55]{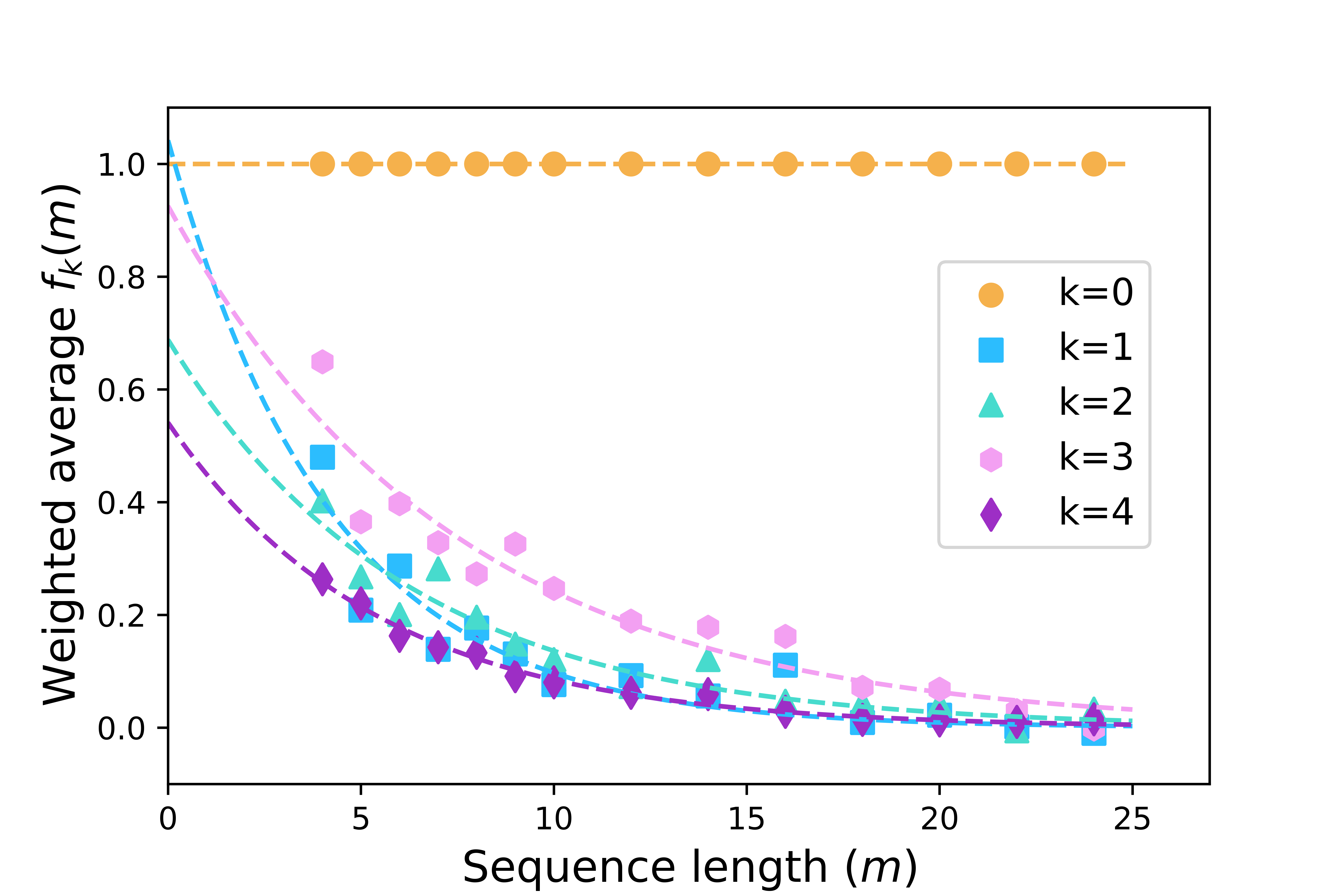}
\caption{Five exponential decays associated with performing matchgate benchmarking using circuits generated by $XY(\theta)$ gates. For readability each exponential decay is offset vertically by some amount. Based on this data we can conclude a two-qubit matchgate circuit fidelity of $F_{\mathrm{avg}} = {0.88\pm0.02}$.}\label{fig:exp-data}
\end{figure}

%=============================================================================
\section{Extensions and variations of the gateset}
%=============================================================================
The benchmarking procedure given in \cref{prot:protocol} for nearest-neighbor generalized matchgate circuits on a line can readily be adapted to related gatesets with interesting properties.

First, following Knill~\cite{knill2001fermionic} (see also \cite{jozsa2013jordan}), one can extend the generalized matchgates by arbitrary single-qubit gates on the first qubit on the line. This corresponds to gates generated by Hamiltonians composed of linear Majorana terms, i.e. $H_{\mathrm{lin}} = \sum_{i}v_i\gamma_i$. Note that the distinction between matchgates and generalized matchgates collapses in this case.
Equivalently, one can add rotations along the $ZX$ axis between the first two neighboring qubits (a cross-resonance gate~\cite{magesan2020effective,sheldon2016procedure}).
This extension corresponds to the group $\SO(2n+1)$~\cite{knill2001fermionic}, which has the spaces~$\Gamma_0$ and~$\Gamma_{2k'-1} \oplus \Gamma_{2k'}$ for $k'\in[n]$ as mutually inequivalent irreducible subspaces.
Matchgate benchmarking is easily adapted by using the correlation functions $\alpha'_{k'}(x,R) \propto \tr\bigl(E_x (P_{2{k'}-1} \!+ \!P_{2k'})( U(R)\rho_0 U(R)\ct) \bigr)$ with $R \in \SO(2n+1)$.
Assuming gate-independent noise, one finds that $f_{k'}(m) = A_{k'}\lambda_{k'}^m$, where the $n+1$ parameters~$\lambda_{k'}\in\mathbb R$ depend only on the noisy implementation of the circuits, and the average gate fidelity can be extracted as $2^{-n} \sum_{k'=0}^{n}\binom{2n+1}{k'} \lambda_{k'} = (2^n + 1)F_{\mathrm{avg}}(\Lambda)-1$.

Second, we can also extend the nearest-neighbor generalized matchgate circuits on a line to those on a circle. This corresponds to gates generated by Hamiltonians of the form~$H_{\mathrm{circle}} = \sum_{i,j}\alpha_{i,j} \gamma[\{i,j\}] + \sum_{i,j} \beta_{i,j} \gamma[\{i,j\}^c]$, together with a single $X$ gate. Their classical simulability was, to our knowledge, first noted in~\cite{brod2014computational}.
One can again work out the corresponding representation theory and write down appropriate correlation functions.

Third, one can also perform matchgate benchmarking with the ordinary matchgate group $\mc{M}_n$ (without the additional bit-flip gate).
Now the representations $\Gamma_k, \Gamma_{2n-k}$ for $k<n$ become equivalent, and the representation~$\Gamma_n$ splits into two inequivalent representations.
The correlation functions change to $\alpha''_k(x,R) \propto\tr\bigl(E_x (P_{k} \!+ \!P_{2n-k})( U(R)\rho_0 U(R)\ct) \bigr)$.
Due to the equivalence of representations, the data $f_k(m)$ for $k\in[n]$ must be fitted to a $2\times 2$ \emph{matrix}-exponential decay $\tr(A_kM_k^m)$, with the eigenvalues of the matrices $M_k$ carrying the fidelity information. This is a significantly harder fitting problem in practice.

Finally, we note that the (orthogonal) matchgate group can be conjugated by a Clifford operator (i.e. consider gates of the form $CU(Q)C\ct$ where $U(Q)\in \mc{M}_n$ and $C$ a Clifford operator) without losing classical simulability~\cite{jozsa2008matchgates}. This conjugation leaves the representation structure, and hence the matchgate benchmarking protocol, unchanged (apart from an appropriate change of initial states and measurement basis). One can for instance consider the orthogonal matchgate group rotated by single qubit Hadamard gates (on each qubit). This rotated orthogonal matchgate group is generated by nearest neighbor $ZZ$ rotations, single qubit $X$ rotation, and $Z$ phase flips. As the $ZZ$ interaction is natural in superconducting circuit qubits and often used to generate two qubit gates~\cite{dicarlo2009demonstration,long2021universal}, this potentially extends the usefulness of matchgate benchmarking.
 % This protocol can also be extended to matchgates (without bit flips) on the circle, corresponding to the action of $\SO(2n)\times \SO(2n)$, under which the representations $(I\!\pm \!Z\tn{n})\Gamma_k\oplus\Gamma_{2n-k}$ are irreducible

\begin{acknowledgments}
% \smallskip\emph{Acknowledgments.}
We would like to acknowledge Harold Nieuwboer, Sergii Strelchuk, Ingo Roth, and Emilio Onorati for useful conversations.
MW acknowledges support by an NWO Veni Innovational Research Grant no.~680-47-459, NWO grant~OCENW.KLEIN.267, and the Deutsche Forschungsgemeinschaft (DFG, German Research Foundation) under Germany's Excellence Strategy - EXC~2092~CASA - 390781972.
This material is based upon work supported by the Defense Advanced Research Projects Agency (DARPA) under agreement No.~HR00112090058.
While preparing this manuscript, we became aware of~\cite{claes2020character}, where a similar procedure for the standard matchgate group is proposed.
\end{acknowledgments}

\bibliographystyle{unsrtnat}
\bibliography{simfermbib}

%=============================================================================
\onecolumn
\clearpage
\appendix
%=============================================================================
\section{Proof of the variance bound}\label{appsec:variance}
%=============================================================================
In this section we give proofs of various technical claims made in the main text.
Consider the representation~$\omega$ of the group~$\OR(2n)$ on the vector space $\Gamma$ of $n$-qubit linear operators, given by $\omega(Q)(\rho) = U(Q) \rho U(Q)^\dagger$ for $Q\in\OR(2n)$ and $\rho\in\Gamma$. This corresponds to the conjugate action of the generalized matchgate group.
Recall that the Majorana product operators $\gamma[S]$ for $S \subseteq [2n]$ form a basis of the space $\Gamma$, and that we defined the subspaces $\Gamma_k = \langle \gamma[S] \;|\; S\subseteq [2n],\; \lvert S\rvert=k\rangle \subseteq \Gamma$ for $k\in\{0,1,\dots,2n\}$.
It is clear from \cref{eq:match_to_orth} that the subspaces $\Gamma_k \subseteq \Gamma$ are invariant, i.e., that $\omega(Q)(\Gamma_k) \subseteq \Gamma_k$. Hence we can consider the restrictions $\omega_k$ of $\omega$ to $\Gamma_k$.
\begin{lemma}\label{lem:rep_th}
  The representation $\omega$ of $\OR(2n)$ decomposes as a direct sum of $2n$ irreducible subrepresentations $\omega_k$ for~$k\in \{0,\ldots 2n\}$, which are all inequivalent.
\end{lemma}
\begin{proof}
  We only need to prove that the representations $\omega_k$ are irreducible and inequivalent.
  Consider the linear map $\Phi_k \colon \Gamma_k \mapsto \wedge^k \mbb C^{2n}$ that sends each $\gamma[S]$ for $S \subseteq [2n]$ to the antisymmetric tensor product~$\wedge_{s \in S} \ket s$.
  This is an isomorphism and moreover $\Phi_k(\omega(Q) \rho) = Q^{\otimes k} \Phi_k(\rho)$, as follows from \cref{eq:match_to_orth}.
  Thus we can infer the irreducibility and mutual inequivalence of $\omega_k$ from the representation theory of $\OR(2n)$ on anti-symmetric tensor powers, which is well-known (see, e.g.,~\cite[Cor.~5.5.6.]{goodman2009symmetry}).
\end{proof}

We now give proof of our variance bound in the noise-free limit.
\begin{theoremone}[restated]
Consider the estimator $\hat{f}_k(m)$ for the quantity $f_k(m)$ defined in \cref{eq:av_data,eq:av_emp}.
Assuming no noise, its variance is bounded (uniformly in $k$ and $m$) as
\[
\md{V}(\hat{f}_k(m)) = \frac{1}{K} O(\mathrm{poly}(n)).
\]
\end{theoremone}
\begin{proof}
Begin by considering general correlation functions $\alpha_k(x,Q) = N_k^{-1}\tr(E_xP_k\omega(Q)(\rho_0))$.
Consider the estimator $\hat{f}_k(m)$ obtained by performing \cref{prot:protocol} for a fixed sequence length $m$, sampling $K$ random sequences and performing $L$ measurements per sequence.
For any fixed $Q\in \OR(2n)$, let $X_k(Q)$ be a random variable taking value $\alpha_k(x,Q)$ with probability $p(x|Q,m)$.
Also let $X_k^{\{L\}}(Q) = \frac{1}{L} \sum_{i=1}^L X^i_k(Q)$ denote the random variable defined by averaging $L$ i.i.d.\ copies $X^i_k(Q)$ of $X_k(Q)$.
Finally let $Y_k$ be the random variable defined by drawing $Q$ uniformly at random and taking the corresponding random variable $X_k^{\{L\}}(Q)$ (that is, $Y_k = X_k^{\{L\}}(Q)$ where $Q$ is uniformly random).
It is clear that the mean of $Y_k$ is $f_k(m)$ ,and moreover the variance of the estimator $\hat{f}_k$ is equal to $\frac{1}{K} \md{V}(Y_k)$.
The variance of $Y_k$ is (by the law of total variation):
\begin{align}
% \md{V}(Y_k)  &= \md{V}_Q \big[ \md{E}(X_k^{\{L\}}(Q))] + \int_{\OR(2n)} dQ\, \md{V}\big[ X_k^{\{L\}}(Q)\big],
\md{V}(Y_k)  &= \md{V}_Q \bigl[ \md{E}\bigl(X_k^{\{L\}}(Q)\bigr) \bigr] + \md{E}_Q \bigl[ \md{V}\bigl( X_k^{\{L\}}(Q) \bigr) \bigr],
\end{align}
where the inner mean and variance are computed for arbitrary but fixed $Q$, while the outer ones are computed with respect to the uniformly random choice of $Q$. Now using the definitions of variance and expectation we get
\begin{align}
\md{V}(Y_k) &= \int_{\OR(2n)} dQ \bigg(\sum_{x\in \{0,1\}^n} \alpha_k(x, Q) p(x|Q,m)\bigg)^2 - f_k(m)^2 \\&\hspace{3em}+ \frac{1}{L} \bigg[\int_{\OR(2n)} dQ \Big(\sum_{x\in \{0,1\}^n} \alpha_k(x, Q)^2 p(x|Q,m)\Big) -\int_{\OR(2n)} dQ\, \Big(\sum_{x\in \{0,1\}^n} \alpha_k(x, Q) p(x|Q,m)\Big)^2\bigg],
\end{align}

Throwing away the negative terms we get
\begin{equation}
\md{V}(Y_k)\leq \frac{L-1}{L} \int_{\OR(2n)} dQ\,\sum_{x,x'\in \{0,1\}^n} \alpha_k(x, Q)\alpha_k(x', Q) p(x|Q,m) p(x'|Q,m) + \frac{1}{L}\sum_{x\in \{0,1\}^n} \alpha_k(x, Q)^2 p(x|Q,m).
\end{equation}
Now using the fact that the $\alpha_k(x,Q)$ are real functions and that for all $a,b\in \mathbb{R}$ we have $2ab\leq a^2 + b^2$ we can simplify this further to
\begin{equation}
\md{V}(Y_k)\leq \frac{L-1}{L} 2^n\int_{\OR(2n)} dQ\,\sum_{x\in \{0,1\}^n} \alpha_k(x, Q)^2 p(x|Q,m)^2 + \frac{1}{L}\sum_{x\in \{0,1\}^n} \alpha_k(x, Q)^2 p(x|Q,m).
\end{equation}
Now define $\ket{\theta_x^k} = \ket{x}$ for even $k$ and $\ket{\theta_x^k} = H\tn{n}\ket{x}$ for odd $k$ (where $H$ is the single qubit Hadamard operator).

We begin by noting that for both~$k$ even and $k$ odd there always exists a generalized matchgate $U(Q_x^k)$ s.t.\ $U(Q_x^k)\ket{\theta_0^k} = \ket{\theta_x^k}$. Hence by the invariance of the Haar measure under left multiplication and the fact that $P_k$ commutes with $\omega(Q_x^k)$ we have that
\begin{equation}
\md{V}(Y_k)\leq \frac{L-1}{L} 2^{2n}\int_{\OR(2n)} dQ \alpha_k(0, Q)^2 p(0|Q,m)^2 + 2^{n} \frac{1}{L} \alpha_k(0, Q)^2 p(0|Q,m).
\end{equation}
We can drop the constant factors of $L$ and upper bound the mixed integrals by monomial integrals (by using that $\alpha_k(0, Q)^2 p(0|Q,m)^2\leq (\alpha_k(0, Q)^4 +p(0|Q,m)^4)/2$ :
\begin{align}
\md{V}(Y_k)&\leq 2^{4n}\max
\left[
\int_{\OR(2n)} dQ\, (2^{-n}\alpha_k(0, Q))^4, \;\;\int_{\OR(2n)} dQ\, p(0|Q,m)^4
\right]\\
&\hspace{15em}
+ 2^{3n}\max
\left[
\int_{\OR(2n)} dQ\, (2^{-n}\alpha_k(0, Q))^3, \;\;\int_{\OR(2n)} dQ\, p(0|Q,m)^3
\right].
\end{align}
First we consider the integrals over the correlation function $\alpha_k$. Setting $t\in \{3,4\}$ we calculate the integral
\begin{align}
\int_{\OR(2n)} dQ\, (2^{-n}\alpha_k(0, Q))^t &= \frac{2^{-tn}}{N_k^t}\int_{\OR(2n)}\tr\big(\dens{\theta^k_0} P_k(U(Q)\dens{\theta^k_0}U(Q)\ct)\big)^t \\
&=\frac{\binom{2n}{k}^t}{\binom{n}{\lfloor{k/2\rfloor}}^{2t}}\int_{\OR(2n)}dQ\, \int_{\OR(2n)}dQ'\,\tr\big(\dens{\theta^k_0} P_k(U(Q')U(Q)\dens{\theta^k_0}U(Q)\ct U(Q')\ct)\big)^t \\
&=\frac{\binom{2n}{k}^t}{\binom{n}{\lfloor{k/2\rfloor}}^{2t}}\tr\bigg(\int_{\OR(2n)}dQ'\, U(Q')\dens{\theta^k_0}U(Q')\ct P_k\bigg(\int_{\OR(2n)}dQ\, U(Q)\dens{\theta^k_0}U(Q)\ct\bigg)\bigg)^t,
\end{align}
using the invariance of the Haar measure and the fact that $P_k$ commutes with the conjugate action $\omega(Q)$. Now, since $P_k$ is an orthogonal projector, it is a contraction in the Hilbert-Schmidt norm, and we get (using Cauchy-Schwartz)
\begin{align}
 \int_{\OR(2n)} dQ\, (2^{-n}\alpha_k(0, Q))^t &\leq\frac{\binom{2n}{k}^t}{\binom{n}{\lfloor{k/2\rfloor}}^{2t}}\tr\bigg(\int_{\OR(2n)}dQ'U(Q')\dens{\theta^k_0}U(Q')\ct \int_{\OR(2n)}dQ\, U(Q)\dens{\theta^k_0}U(Q)\ct \bigg)^t.
 \end{align}
 Using the invariance of the Haar measure (to absorb one of the integrals) and the definition of $p(0|Q,m)$ we see that the RHS becomes
 \begin{align}
 \frac{\binom{2n}{k}^t}{\binom{n}{\lfloor{k/2\rfloor}}^{2t}}&\tr\bigg(\int_{\OR(2n)}dQ'U(Q')\dens{\theta^k_0}U(Q')\ct \int_{\OR(2n)}dQ\, U(Q)\dens{\theta^k_0}U(Q)\ct \bigg)^t\\
 &= \frac{\binom{2n}{k}^t}{\binom{n}{\lfloor{k/2\rfloor}}^{2t}}\int_{\OR(2n)}dQ\, \tr\big(\dens{\theta^k_0})U(Q)\dens{\theta^k_0}U(Q)\ct \big)^t\\
 &= \frac{\binom{2n}{k}^t}{\binom{n}{\lfloor{k/2\rfloor}}^{2t}}\int_{\OR(2n)}dQ\, p(0|Q,m)^t,
 \end{align}
 and hence
 \begin{align}
 \int_{\OR(2n)} dQ\, (2^{-n}\alpha_k(0, Q))^t \leq \frac{\binom{2n}{k}^t}{\binom{n}{\lfloor{k/2\rfloor}}^{2t}}\int_{\OR(2n)}dQ\, p(0|Q,m)^t.
 \end{align}
 Now we use Stirling's approximation ($\sqrt{2\pi n} (n/e)^n e^{\frac 1{12n+1}}\leq n! \leq \sqrt{2\pi n} (n/e)^n e^{\frac 1{12n}}$) to note that $\frac{\binom{2n}{k}}{\binom{n}{\lfloor{k/2\rfloor}}^2} = O\big(\mathrm{poly}(n,k)\big)$ we thus only need to consider the integral over $p(0|Q,m)$:
\begin{align}
\int_{\OR(2n)} dQ\, p(0|Q,m)^t &= \int_{\OR(2n)} dQ \tr\big(\dens{\theta_0^k} U(Q) \dens{\theta_0^k}U(Q)\ct\big)^t\\
 &= \int_{\SO(2n)} dR \tr\big(\dens{\theta_0^k} U(R) \dens{\theta_0^k}U(R)\ct\big)^t \\
 &\hspace{10em}+ \int_{\SO(2n)} dR \tr\big(\dens{\theta_0^k}U(R)X_n \dens{\theta_0^k}X_nU(R)\ct\big)^t,
\end{align}
where $X_n$ is a bit flip on the last qubit. Defining $\ket{e} = X_n\ket{0}\tn{n}$ and noting that $X_n\ket{+}^{\ot n} = \ket{+}^{\ot n}$ we see that for both~$k$ even and $k$ odd:
\begin{equation}
\int_{\OR(2n)} dQ\, p(0|Q,m)^t \leq 2 \max \bigg\{\int_{\SO(2n)} dR \tr\big(\dens{\phi} U(R) \dens{\phi}U(R)\ct\big)^t \;\;\bigg|\;\; \ket{\phi} \in \{\ket{0},\ket{e},\frac{1}{\sqrt{2}}(\ket{0}+\ket{e})\}\bigg\},
\end{equation}
where we used that $\ket{+}^{\ot n} = U(R_+) \frac{1}{\sqrt{2}}(\ket{0}+\ket{e})$ for some $R_+\in \SO(2n)$ as well as Haar invariance. All these integrals are bounded by the moment bounds in \cref{lem:integral_bounds}, and we find that
\begin{equation}
 \int_{\OR(2n)} dQ\, p(0|Q,m)^t = 2^{-tn} \, O(\mathrm{poly}(n)).
\end{equation}
Hence we have
\begin{align}
\md{V}(Y_k)= O(\mathrm{poly}(n)),
\end{align}
for all $k \in [2n]$, which proves the theorem.
\end{proof}

%=============================================================================
\section{Proof of the moment bound}\label{appsec:moments}
%=============================================================================
In this section we state and prove \cref{lem:integral_bounds}, which contains our moment bounds and implies \cref{lem:moment_bound} in the main text.
The $n$-qubit Hilbert space $(\mbb C^2)^{\ot n}$ carries a representation of the Lie algebra~$\so(2n)$.
The corresponding Lie group representation corresponds precisely to the matchgate action, so that \cref{eq:match_to_orth} holds.
We now recall some notions from representation theory (see, e.g.,~\cite{fulton2013representation} for a gentle introduction).
Define fermionic creation and annihilation operators by
\begin{align}
  a_i = \frac{\gamma_{2i-1} + i \gamma_{2i}}2, \quad
  a_i^\dagger = \frac{\gamma_{2i-1} - i \gamma_{2i}}2
  \quad \text{ for } i \in [n],
\end{align}
in terms of the Majorana fermion operators, which were defined in the main text as operators on $(\mbb C^2)^{\ot n}$.
In this way, we can identify $(\mbb C^2)^{\ot n}$ with the fermionic Fock space $\bigwedge \mbb C^n$.
Note that all-zero basis state $\ket0$ corresponds to the fermionic Fock vacuum, since we have $a_i \ket{0} = 0$ for $i\in[n]$.
Now, consider the following operators:
\begin{align}\label{eq:gengen}
   X'_{ij} = a_i^\dagger a_j - \frac12 \delta_{ij}, \quad
   Y'_{ij} = a_i^\dagger a_j^\dagger \;\; \text{(for $i<j$)}, \quad
   Z'_{ij} = a_i a_j \;\; \text{(for $i<j$)}.
\end{align}
Here we follow the notation and conventions of~\cite{fulton2013representation} (and caution that these operators are not the Pauli matrices).
The operators defined in \cref{eq:gengen} satisfy the commutation relations of $\so(2n)$, and hence define a representation of~$\so(2n)$ on the $n$-qubit Hilbert space.
This representation generates the matchgate group.

Any irreducible representation of $\so(2n)$ is classified by its highest weight.
A \emph{weight vector} is a joint eigenvector of the operators $ H'_i := X'_{ii}$ for $i\in[n]$; and the vector of eigenvalues is simply called the \emph{weight}.
Note that, as an operator on the $n$-qubit Hilbert space, $ H'_i$ is nothing but $-\frac12Z_i$ ($Z_i$ is the Pauli $Z$-matrix acting on the $i$-th qubit).
Moreover, a \emph{highest weight vector} is a weight vector that is annihilated by the operators $\{ X'_{ij}\}_{i<j}$ and $\{ Y'_{ij}\}_{i<j}$.
Any irreducible representation contains a unique (up to phase) highest weight vector, and its weight, called the \emph{highest weight}, characterizes the representation completely.
It is well-known that the representation of $\so(2n)$ on $(\mbb C^2)^{\ot n}$ defined above decomposes into two irreducible representations, with highest weights
\begin{align}
  \alpha := (\tfrac12,\dots,\tfrac12,\tfrac12) \quad\text{and}\quad
  \beta := (\tfrac12,\dots,\tfrac12,-\tfrac12) \in (\mbb Z/2)^n.
\end{align}
If $n$ is even, the subrepresentation with highest weight~$\alpha$ is the even particle number subspace of $(\mbb C^2)^{\ot n}$, and the subrepresentation with highest weight~$\beta$ is the odd particle number subspace.
If $n$ is odd, then the reverse is true: $\beta$ corresponds to the even particle number subspace, and $\alpha$ corresponds  to the odd particle number subspace.
See~\cite[Prop.~20.15]{fulton2013representation}.
Indeed, the all-one state $\ket1^{\ot n}$ is a highest weight vector of weight $\alpha$ (but its parity depends on the parity of $n$), while the state $X_n \ket 1^{\ot n}$ is a highest weight vector of weight $\beta$.

In our analysis, it will be useful to instead consider the \emph{lowest weight vectors}, which are the weight vectors that are annihilated by the $\{ X'_{ij}\}_{i>j}$ and $\{ Z'_{ij}\}_{i<j}$.
Just like the highest weight vectors, they characterize the irreducible representation uniquely.
Clearly, both the all-zero state $\ket0$ and the state $\ket e = a_n^\dagger \ket0 = X_n \ket 0$ are lowest weight vectors.
We can compute the \emph{highest} weight of the corresponding representation by observing that $\ket0$ is the vacuum (an even particle number state), while $\ket e$ is a single-particle state (an odd particle number state).
We summarize: % our discussion in the following lemma:

\begin{lemma}\label{lem:correspondence}
The vectors $\ket 0$ and $\ket e = a^\dagger_n \ket 0 = X_n \ket 0$ are lowest weight vectors in $(\mbb C^2)^{\ot n}$.
For even $n$, $\ket0$ is contained in an irreducible subrepresentation with highest weight $\alpha$ and $\ket e$ is contained in an irreducible subrepresentation with highest weight $\beta$.
For odd $n$, the opposite holds: $\ket0$ corresponds to highest weight $\beta$ and $\ket1$ to highest weight~$\alpha$.
\end{lemma}

In general, the highest weight can be obtained from the lowest weight by the action of the Weyl group, which can permute the entries of the weight as well as swap an even number of signs, until we obtain a weight $\omega \in (\mbb Z/2)^n$ that satisfies $\omega_1 \geq \dots \geq \omega_{n-1} \geq \lvert \omega_n\rvert$.
For $\ket0$, the weight is $(-\tfrac12,\dots,-\frac12,-\frac12)$, so we obtain $\alpha$ if $n$ is even and $\beta$ if $n$ is odd.
For $\ket e$, the weight is $(-\tfrac12,\dots,-\frac12,\frac12)$, so we obtain $\beta$ if $n$ is even and $\alpha$ if $n$ is odd.

It is well-known and easy to see that the tensor product of highest weight vectors is again a highest weight vector, with associated highest weight the sum of the highest weights of the individual tensor factors (see~\cite[Obs.~13.2]{fulton2013representation}).
The same is true for lowest weight vectors.
Accordingly, for any $t$ and $m\in[t]$, we can define the following lowest weight vector:
\begin{align}\label{eq:when they go low}
  \ket {\Omega^t_m} := \ket 0^{\ot m} \ot \ket e^{\ot t-m}.
\end{align}
The corresponding highest weight $\lambda_m^t \in (\mbb Z/2)^n$, computed as described above, is the following:
\begin{equation}\label{eq:we go high}
  \lambda_m^t := \left\{\begin{array}{cc}
    m\alpha + (t-m)\beta = (\frac t2,\cdots, \frac t2, m -\frac t2)& \text{ if $n$ is even,} \\
    m\beta + (t-m)\alpha=(\frac t2,\cdots, \frac t2, \frac t2-m) & \text{ if $n$ is odd.}
\end{array} \right.
\end{equation}
These arguments will crucially feature in our moment bound.
Indeed, the computation of the $2t$-th moment can be reduced to an integral involving the vectors $\ket{\Omega^t_m}$ for $m\in[t]$.
Since each $\ket{\Omega^t_m}$ is a lowest weight vector, it is supported in a \emph{single} irreducible representation.
Therefore, we can use powerful tools from representation theory such as Schur's lemma and the Weyl dimension formula to compute the corresponding integrals.

Before we can state our moment calculations, we must discuss one subtlety.
Not every representation of the Lie algebra $\so(2n)$ integrates to a representation of the Lie group~$\SO(2n)$.
In particular, this problem occurs for the representation on $\so(2n)$ on $(\mbb C^2)^{\ot n}$ discussed above -- meaning that the unitaries $U(R)$ for $R\in\SO(2n)$ do \emph{not} define a representation of~$\SO(2n)$.
So far, this was not important for our analysis, since the conjugation with $U(R)$ as in \cref{eq:match_to_orth} gives a well-defined representation of $\SO(2n)$.
However, it will be necessary to be mindful of this subtlety in what follows.
Fortunately, for any representation of $\so(2n)$ we always have a corresponding representation of the Lie group~$\Spin(2n)$, which has the same Lie algebra but is a simply connected double cover of $\SO(2n)$.
We denote the $\Spin(2n)$-representation obtained in this way from the $\so(2n)$-representation on $(\mbb C^2)^{\ot n}$ by $\rho(R)$ for~$R\in\Spin(2n)$.
The matchgate group~$\mathcal{M}_n$ is nothing but the image of $\Spin(2n)$ under this representation.

\begin{lemma}
For integers $t$ and complex numbers $x, y$ we have
\begin{align}\label{eq:main_moment}
&\int_{\SO(2n)} dR \,\, \left|\bra 0 (\overline x+ \overline y a_n ) U(R) (x+ya_n^\dagger) \ket 0\right|^{2t}\notag \\ &\hspace{1em}=(n-1)!^2\sum_{m=0}^t  |x|^{4m}|y|^{4(t-m)}  \binom{t}{m}^2 \prod_{i=1}^t \frac{(i+n-1)!}{\sqrt{i!(2n-2+i)!}}\sqrt{\binom{n+t/2-1}{t/2}}\frac{m!}{(m+n-1)!}\frac{(t-m)!}{(t-m+n-1)!}.
\end{align}
\end{lemma}
\begin{proof}
We first write this as a Haar integral over the spin group, which does not change the value as explained above.
\begin{align}
  \int_{\SO(2n)} dR \,\, \left|\bra 0 (\overline x+ \overline y a_n ) U(R) (x+ya_n^\dagger) \ket 0\right|^{2t}
= \int_{\Spin(2n)} dR \,\, \left|\bra 0 (\overline x+ \overline y a_n ) \rho(R) (x+ya_n^\dagger) \ket 0\right|^{2t}.
\end{align}
Recall that $\ket e = a_n^\dagger \ket 0$.
The Gaussian fermionic unitaries preserve the parity of fermions, therefore,
\begin{align}\label{eq:dskfh}
\int_{\SO(2n)} dR \,\, &\left|\bra 0 (\overline x+\overline y a_n ) \rho(R) (x+ya_n^\dagger) \ket 0\right|^{2t}\\
&=\int_{\SO(2n)} dR \,\, \left||x|^2 \bra 0  \rho(R) \ket 0 + |y|^2 \bra e  \rho(R) \ket e \right|^{2t}\\
&=\int_{\SO(2n)} dR \,\, \left| \sum_{m=0}^t \binom{t}{m}|x|^{2m}|y|^{2(t-m)} \bra 0  \rho(R) \ket 0^m \bra e  \rho(R) \ket e^{t-m} \right|^2 \\
&=\!\!\sum_{m,m'=0}^{t,t}\binom{t}{m}\binom{t}{m'} |x|^{2m+2m'}|y|^{2(t-m)+2(t-m')}\bigg(\int_{\SO(2n)} dR \, \bra 0  \rho(R) \ket 0^m  \bra 0  \rho(R)^\dagger \ket 0^{m'} \\
&\hspace{27em}\times \bra e  \rho(R) \ket e^{t-m} \bra e  \rho(R)^\dagger \ket e^{t-m'} \bigg),
\end{align}
where the first equality  follows from preserving the parity. The rest are simple algebraic manipulations.
Now recall that we defined $\ket {\Omega^t_m} := \ket 0^{\ot m} \ot \ket e^{\ot t-m}$ in \cref{eq:when they go low}.
Thus, \cref{eq:dskfh} can be written as
\begin{align}
\int_{\SO(2n)}& dR \,\, \left|\bra 0 (\overline x+\overline y a_n ) \rho(R) (x+ya_n^\dagger) \ket 0\right|^{2t} \notag\\
&=\!\!\!\!\sum_{m,m'=0}^{t,t}\binom{t}{m}\binom{t}{m'} |x|^{2m+2m'}|y|^{2(t-m)+2(t-m')}\left(\int_{\Spin(2n)} dR \, \bra{\Omega^t_m} \rho^{\ot t}(R) \ket{\Omega^t_m} \bra{\Omega^t_{m'}} \rho^{\ot t}(R)^\dagger \ket{\Omega^t_{m'}}\right).
\end{align}
Now, because $\ket {\Omega^t_m}$ is a lowest weight vector, its support is limited to the irreducible representation of the $\text{so}(2n)$ Lie algebra given by the corresponding highest weight $\lambda_m^t$ as given in \cref{eq:we go high}. We denote this irreducible representation by $\rho_{\lambda_m^t}$.
Therefore, the above relation can be re-written as
\begin{align}
\int_{\SO(2n)}& dR \,\, \left|\bra 0 (\overline x+\overline y a_n ) \rho(R) (x+ya_n^\dagger) \ket 0\right|^{2t} \notag\\
&=\!\!\!\!\sum_{m,m'=0}^{t,t}\binom{t}{m}\binom{t}{m'} |x|^{2m+2m'}|y|^{2(t-m)+2(t-m')}\left(\int_{\Spin(2n)} dR \, \bra{\Omega^t_m} \rho_{\lambda_m^t}(R) \ket{\Omega^t_m} \bra{\Omega^t_{m'}} \rho_{\lambda_{m'}}(R)^\dagger \ket{\Omega^t_{m'}}\right).
\end{align}
Let us focus on the integral in the parenthesis.
As a result of the Schur orthogonality relations, it is straightforward to see that the Haar integral $\int{dR\rho_{\lambda_m^t}(R) \ket{\Omega^t_m} \bra{\Omega^t_{m'}} \rho_{\lambda_{m'}}(R)^\dagger }$ vanishes if $m\neq m'$. Furthermore, if $m=m'$, then
\begin{equation}
\int_{\Spin(2n)} {dR\,\rho_{\lambda_m^t}(R) \ket{\Omega^t_m} \bra{\Omega^t_{m}} \rho_{\lambda_{m}}(R)^\dagger }\propto \rho_{\lambda_{m}}(\text{Id}) ,
\end{equation}
as a consequence of Schur's lemma. The proportionality constant can be calculated by comparing the traces of the both side of the equality, and it is equal to $1/\dim ({\rho_{\lambda_m^t}})$. Hence, we obtain,
\begin{equation}
\int_{\Spin(2n)} dR \,\, \left|\bra 0 (\overline x+ \overline y a_n ) \rho(R) (x+ya_n^\dagger) \ket 0\right|^{2t} = \sum_{m=0}^t \binom{t}{m}^2 |x|^{4m}|y|^{4(t-m)}\frac{1}{\dim ({\rho_{\lambda_m^t}})}.
\end{equation}
The dimension of~$\rho_{\lambda_m^t}$ can be directly calculated using explicit relations for the dimension of the irreducible representations of $\so(2n)$. For an arbitrary highest weight $\mu = (\mu_1,\mu_2,\cdots,\mu_n)$, the dimension of the corresponding irreducible representation is given by one of Weyl dimension formulas~\cite[Eq.~24.41]{fulton2013representation}:
\[\dim{(\rho_{\mu})} = \prod_{1\leq i<j\leq n} \frac{l_i^2-l_j^2}{m_i^2-m_j^2},\quad\text{with }m_i:= n-i, \text{ and }l_i:=\mu_i+n-i.\]
After a few lines of algebra, this leads to,
\begin{equation}
\dim{(\rho_{\lambda_m^t})} =\left( \prod_{1\leq i< j \leq n-1}\frac{t+i+j}{i+j}\right)\times \frac{1}{(n-1)!^2} \times \frac{(m+n-1)!}{m!}\times \frac{(t-m+n-1)!}{(t-m)!}.
\end{equation}
We focus on the term inside the parentheses. By some manipulation of the product factors and the definition of the binomial we see:
\begin{align}
\prod_{1\leq i< j \leq n-1}\frac{t+i+j}{i+j} &= \sqrt{ \prod_{1\leq i< j \leq n-1}\frac{t+i+j}{i+j} \times \prod_{1\leq j< i \leq n-1}\frac{t+i+j}{i+j} }\\
&= \sqrt{ \prod_{\substack{i,j=1\\i\neq j}}^{n-1}\frac{t+i+j}{i+j}}\\
&=\sqrt{\prod_{i=1}^{n-1}\frac{i}{i+\frac t2}}\times \sqrt{\prod_{i,j=1}^{ n-1}\frac{t+i+j}{i+j}}\\
&=\binom{\frac t2+n-1}{\frac t2}^{-1/2}\times \sqrt{\prod_{i,j=1}^{ n-1}\frac{t+i+j}{i+j}}.
\end{align}
Furthermore, we see that
\begin{align}
    \prod_{i,j=1}^{ n-1}\frac{t+i+j}{i+j} &= \prod_{i=1}^{n-1}\left(\frac{\prod_{j=1}^{n-t-1} (t+i+j)}{\prod_{j=t+1}^{n-1} (i+j)}\times \frac{\prod_{j=n-t}^{n-1} (t+i+j)}{\prod_{j=1}^{t} (i+j)} \right)
\end{align}
which is nothing more than splitting the products over the index $j$ into two components set by $t$. Working this out further we get
\begin{align}
    \prod_{i,j=1}^{ n-1}\frac{t+i+j}{i+j}     &=\prod_{i=1}^{n-1}\left(\frac{\prod_{j=1}^{n-t-1} (t+i+j)}{\prod_{j=1}^{n-1-t} (t+i+j)}\times \frac{\prod_{j=1}^{t} (n-1+i+j)}{\prod_{j=1}^{t} (i+j)} \right)\\
     &=\prod_{i=1}^{n-1}\left( \frac{\prod_{j=1}^{t} (n-1+i+j)}{\prod_{j=1}^{t} (i+j)} \right)\\
     &=\prod_{i=1}^{t}\left( \frac{\prod_{j=1}^{n-1} (n-1+i+j)}{\prod_{j=1}^{n-1} (i+j)} \right)\\
     &=\prod_{i=1}^{t} i! \frac{(2n-2+i)!}{(i+n-1)!^2}.
\end{align}
Combining all of these relations, we have:
\begin{align}
\int_{\Spin(2n)}& dR \,\, \left|\bra 0 (\overline x+ \overline y a_n ) \rho(R) (x+ya_n^\dagger) \ket 0\right|^{2t} \\ &=(n-1)!^2\sum_{m=0}^t  |x|^{4m}|y|^{4(t-m)}  \binom{t}{m}^2 \prod_{i=1}^t \frac{(i+n-1)!}{\sqrt{i!(2n-2+i)!}}\sqrt{\binom{n+t/2-1}{t/2}}\frac{m!}{(m+n-1)!}\frac{(t-m)!}{(t-m+n-1)!}.\notag
\end{align}
\end{proof}
Lastly, we prove the main result of this section.
\begin{theorem}\label{lem:integral_bounds}
Let $\ket{0}$ be the all-zero state and define the state $\ket{e} = X_n\ket{0} = a_n\ct \ket{0}$. We have, for any fixed $t$:
\begin{align}
&\int_{\SO(2n)} dR \,\, \left|\frac{\bra{0}+\bra{e}}{\sqrt{2}} U(R) \frac{\ket{0}+\ket{e}}{\sqrt{2}}\right|^{2t} = 2^{-tn} \, O(\mathrm{poly}(n)),\\
&\int_{\SO(2n)} dR \,\, \left|\bra 0  U(R)  \ket 0\right|^{2t} = 2^{-tn} \, O(\mathrm{poly}(n)),\label{eq:secondone}\\
&\int_{\SO(2n)} dR \,\, \left|\bra e U(R)   \ket e\right|^{2t} = 2^{-tn} \, O(\mathrm{poly}(n)).\label{eq:thirdone}
\end{align}
\end{theorem}
\begin{proof}
We note that the first integral is of the form \cref{eq:main_moment} with $x=\frac1{\sqrt 2}$ and $y=\frac 1{\sqrt 2}$:
\begin{align}
\int_{\SO(2n)} dR \,\, &\left|\frac{\bra{0}+\bra{e}}{\sqrt{2}} U(R) \frac{\ket{0}+\ket{e}}{\sqrt{2}}\right|^{2t} \\
&=4^{-t}(n-1)!^2\sum_{m=0}^t  \binom{t}{m}^2 \prod_{i=1}^t \frac{(i+n-1)!}{\sqrt{i!(2n-2+i)!}}\sqrt{\binom{n+t/2-1}{t/2}}\frac{m!}{(m+n-1)!}\frac{(t-m)!}{(t-m+n-1)!}\\
&=4^{-t}(n-1)!^2t!^2 \prod_{i=1}^t \frac{(i+n-1)!}{\sqrt{i!(2n-2+i)!}}\sqrt{\binom{n+t/2-1}{t/2}}\sum_{m=0}^t  \frac1{m!(t-m)!}\frac{1}{(m+n-1)!}\frac{1}{(t-m+n-1)!}\notag\\
&=\frac{4^{-t}(n-1)!^2t!^2}{(t+n-1)!^2} \prod_{i=1}^t \frac{(i+n-1)!}{\sqrt{i!(2n-2+i)!}}\sqrt{\binom{n+t/2-1}{t/2}}\sum_{m=0}^t  \binom{t+n-1}{m}\binom{t+n-1}{t-m}.
\end{align}
Now, we can use the binomial identity
\begin{equation}
\sum_{m=0}^t  \binom{t+n-1}{m}\binom{t+n-1}{t-m} = \binom{2(t+n-1)}{t}.
\end{equation}
Hence, we have
\begin{align}\label{eq:mid_bound}
    \int_{\SO(2n)} dR \,\, &\left|\frac{\bra{0}+\bra{e}}{\sqrt{2}} U(R) \frac{\ket{0}+\ket{e}}{\sqrt{2}}\right|^{2t} \\
    &=4^{-t} \prod_{i=1}^t \frac{(i+n-1)!}{\sqrt{i!(2n-2+i)!}}\sqrt{\binom{n+t/2-1}{t/2}} \binom{2(t+n-1)}{t}\binom{t+n-1}{t}^{-2}\\
    &=4^{-t}  \left(\prod_{i=1}^t\binom{2n+2i-2}{n+i-1}^{-1/2}\right) \left(\prod_{i=1}^t\binom{2n+2i-2}{i}^{1/2}\right)\sqrt{\binom{n+t/2 -1 }{t/2}}\binom{2(t+n-1)}{t}\binom{t+n-1}{t}^{-2}.\label{eq:llll}
\end{align}
The last four terms in~\cref{eq:llll} (up to square roots) are polynomials in $n$. Therefore, we have
\begin{align}\label{eq:mid_bound2}
    \int_{\SO(2n)} dR \,\, \left|\frac{\bra{0}+\bra{e}}{\sqrt{2}} U(R) \frac{\ket{0}+\ket{e}}{\sqrt{2}}\right|^{2t}
    =4^{-t}  \left(\prod_{i=1}^t\binom{2n+2i-2}{n+i-1}^{-1/2}\right) \times O(\mathrm{poly}(n)).
\end{align}
% \mw{Why do we need the lower bounds in the following?}
Using Stirling's approximation $\sqrt{2\pi n} (n/e)^n e^{\frac 1{12n+1}}\leq n! \leq \sqrt{2\pi n} (n/e)^n e^{\frac 1{12n}}$, the following bound holds,
\begin{equation}\label{eq:centralbinom}
\frac{4^{n}}{\sqrt{\pi n}}\exp\left[-\frac1{8n}\right]\leq\frac{4^{n}}{\sqrt{\pi n}}\exp\left[\frac{1}{24n+1} \!-\! \frac{1}{6n}\right]\leq \binom{2n}{n}.
% \leq \frac{4^{n}}{\sqrt{\pi n}}\exp\left[\frac{1}{24n} \! -\!\frac{2}{12n+1}\right] \leq  \frac{4^{n}}{\sqrt{\pi n}}\exp\left[\frac{1}{72n^2}\!-\!\frac1{8n}\right]\!\!,\!\!\!\!
\end{equation}
We can apply this lower bound  to every factor of $\binom{2n+2i-2}{n+i-1}$ in the product factor of \cref{eq:mid_bound2} to obtain an upper bound of this factor:
\begin{align}
% 2^{-tn}2^{-t^2-2t}\pi^{t/4}\left(\frac{(n+t-1)!}{(n-1)!}\right)^{1/4}& \exp\left(\frac{t}{16(n+t)}\right)\\
% \leq  
4^{-t}&\left(\prod_{i=1}^t\binom{2n+2i-2}{n+i-1}^{-1/2}\right)
% &\hspace{2em}
\leq 2^{-tn}2^{-t^2-2t}\pi^{t/4}\left(\frac{(n+t-1)!}{(n-1)!}\right)^{1/4}\exp\left(\frac{t}{16n}\right).
\end{align}
Inserting this relation back into \cref{eq:mid_bound2}, we conclude that
\begin{equation}
 \int_{\SO(2n)} dR \,\, \left|\frac{\bra{0}+\bra{e}}{\sqrt{2}} U(R) \frac{\ket{0}+\ket{e}}{\sqrt{2}}\right|^{2t}  =2^{-tn}\time \times O(\mathrm{poly}(n)).
\end{equation}
\par
Next, we discuss~\cref{eq:secondone,eq:thirdone}. Using~\cref{eq:main_moment} for $x=0, y=1$ and $x=1, y=0$ we have
\begin{align}\label{eq:lastp}
    \int_{\SO(2n)} dR \,\, \left|\bra 0  U(R)  \ket 0\right|^{2t}=&\int_{\SO(2n)} dR \,\, \left|\bra e  U(R)  \ket e\right|^{2t}  \notag\\
    =&\left(\prod_{i=1}^t \frac{(i+n-1)!}{\sqrt{i!(2n-2+i)!}}\right)\sqrt{\binom{n+t/2-1}{t/2}}\frac{t!}{(t+n-1)!} \notag\\
    =&\left(\binom{2n-2}{n-1}^{-1/2}\right)^t \times \left(\prod_{i=1}^t \binom{i+n-1}{i}\binom{2n-2+i}{2n-2}^{-1/2} \right)\notag\\
    =&\left(\binom{2n-2}{n-1}^{-1/2}\right)^t \times O(\mathrm{poly}(n)).
\end{align}
% Again, we suppress the dependence of $O(\mathrm{poly}(n))$ on $t$, as we are solely interested in the dependence on $n$.
Incorporating~\cref{eq:centralbinom} into~\cref{eq:lastp} we immediately obtain the desired results~\cref{eq:secondone,eq:thirdone}.
\end{proof}

We note that \cref{lem:moment_bound} follows from \cref{lem:integral_bounds}.
For the all-plus state, this follows by Haar invariance and the fact that $\ket{+}^{\ot n} = U(R_+) \frac{1}{\sqrt{2}}(\ket{0}+\ket{e})$ for some $R_+\in \SO(2n)$.

%=============================================================================
\section{Derivation of decay model in the gate-independent noise case}\label{appsec:gate_indep}
%=============================================================================
In this section we provide an ``artisanal'' derivation of \cref{eq:fid} under the assumption of gate-independent noise.
We assume that there exists a quantum channel~$\Lambda$ such that every generalized matchgate~$U(Q)$ is implemented on the device as~$\Lambda \circ \omega(Q)$, where $\omega(Q)(\rho) = U(Q)\rho \,U(Q)\ct$ as defined above.
We note that this assumption is not very realistic and the decay model can be derived under much weaker conditions using the general arguments given in~\cite{framework} (specifically Theorem~$9$ therein). However this derivation has the benefit of making it more clear what all the moving parts are. We will give an argument that the gate-dependent noise assumption can be relaxed in the next section. \\

Throughout the derivation we will make use of the matrix-transfer representation, writing matrices $\rho,E$ as vectors $\ket{\rho},\bra{E}$ with trace inner product $\langle E|\rho\rangle = \tr(E\ct \rho)$. Correspondingly, superoperators $\Lambda$ get mapped to matrices acting as $\Lambda\ket{\rho} = \ket{\Lambda(\rho)}$. This representation maps composition to matrix multiplication and interplays correctly with tensor products. In this picture we can write the output of the matchgate benchmarking protocol as
\begin{align}
f_k(m) &= \int_{\OR(2n)} dQ_1\cdots Q_m\sum_{x\in \{0,1\}^{n}}\alpha_k(x,Q_1\cdots Q_m) \bra{\tilde{E}_x} \Lambda \omega(Q_m) \cdots \Lambda \omega(Q_1) \ket{\tilde{\rho}_0}\\
&= N_k^{-1}\int_{\OR(2n)} dQ_1\cdots Q_m\sum_{x\in \{0,1\}^{n}} \bra{E_x\otimes \tilde{E}_x}(P_k\otimes \Lambda)\omega(Q_m)\tn{2} (\id \otimes \Lambda) \cdots  (\id \otimes \Lambda) \omega(Q_1)\tn{2} \ket{\rho_0\otimes \tilde{\rho}_0},
\end{align}
with $\tilde{E}_x,\tilde{\rho}_0$ noisy versions of their ideal counterparts.  Using standard representation theory and the decomposition of \cref{lem:rep_th} we can express the integral $\int_{\OR(2n)} dQ \omega(Q)\tn{2}$ as a projector onto a space spanned by the vectors
\begin{equation}
\{ \ket{\mathrm{v}(P_{k'})}\;\;\;|\;\;k\in \{0,\ldots,2n\}\},
\end{equation}
where
\begin{equation}
\ket{\mathrm{v}(P_{k'})} := 2^{-n}\sum_{S\subset[2n],|S|=k'}\ket{\gamma[S]\otimes\gamma[S]},
\end{equation}
with $P_{k'} = 2^{-n}\sum_{S\subset[2n],|S|=k'}\ket{\gamma[S]}\bra{\gamma[S]}$ the projector onto the subrepresentation $\Gamma_k$.  Using linearity, and writing $|P_{k'}|$ for the dimension of $\Gamma_{k'}$, we can insert this into the expression for $f_k(m)$

\begin{align}
f_k(m) &= N_k^{-1} \sum_{x\in \{0,1\}^{n}} \bra{E_x\otimes \tilde{E}_x}
(P_k\otimes \Lambda)
\bigg[
\Big[
\sum_{k'=0}^{2n} \frac{1}{|P_{k'}|}\ket{\mathrm{v}(P_{k'})} \bra{\mathrm{v}(P_{k'})}
\Big]
(\id \otimes \Lambda)
\Big[
\sum_{k'=0}^{2n} \frac{1}{|P_{k'}|} \ket{\mathrm{v}(P_{k'})} \bra{\mathrm{v}(P_{k'})}
\Big]
\bigg]^m
\ket{\rho_0\otimes \tilde{\rho}_0}.\notag
\end{align}
% From now on we will assume $k<n$. The $k=n$ case is somewhat different and we treat it separately later.
Now using the property $A\otimes \id \ket{\mathrm{v}(B)} = \ket{\mathrm{v}(BA)}$ of the vectorization function and the orthogonality of the projectors $P_k$ we can rewrite this as
\begin{align}
f_k(m) &= N_k^{-1} \sum_{x\in \{0,1\}^{n}}
\bra{E_x\otimes \Lambda\ct(\tilde{E}_x)}
\frac{1}{|P_{k'}|}
\ket{\mathrm{v}(P_{k})}
\bigg[
\frac{1}{|P_{k'}|}\bra{\mathrm{v}(P_{k})} (\id \otimes \Lambda)\ket{\mathrm{v}(P_{k})}
\bigg]^m
\bra{\mathrm{v}(P_{k})}
\ket{\rho_0\otimes \tilde{\rho}_0}.
\end{align}
% Using another property of the vectorisation map, namely that $\bra{\mathrm{v}(A)}\ket{\mathrm{v}(B)} = \tr(A\ct B)$ (this being the trace of superoperators $A$ and $B$), we can define the matrices
% \begin{align}
% M_k &= \begin{pmatrix}\frac{1}{|P_{k}|}\tr(P_k\Lambda) & \frac{1}{|P_{k}|}\tr(T_k\Lambda)\\
% \frac{1}{|P_{k}|}\tr(T_{2n-k}\Lambda) & \frac{1}{|P_{k}|}\tr(P_{2n-k}\Lambda)\end{pmatrix}\\
% A_k^x &= \frac{1}{|P_{k}|}\bra{E_x\!\otimes\! \Lambda\ct(\tilde{E}_x)}\!\begin{pmatrix}( \ket{\mathrm{v}(P_k)}\bra{ \mathrm{v}(P_k)} +\ket{\mathrm{v}(T_{2n-k})}\bra{ \mathrm{v}(T_{2n-k})} )  &( \ket{\mathrm{v}(P_{2n-k})}\bra{ \mathrm{v}(T_k)} +\ket{\mathrm{v}(T_{2n-k})}\bra{ \mathrm{v}(P_{k})} )
% \\( \ket{\mathrm{v}(P_{k})}\bra{ \mathrm{v}(T_{2n-k})} +\ket{\mathrm{v}(T_{k})}\bra{ \mathrm{v}(P_{2n-k})} ) &( \ket{\mathrm{v}(P_{2n-k})}\bra{ \mathrm{v}(P_{2n-k})} +\ket{\mathrm{v}(T_{k})}\bra{ \mathrm{v}(T_{k})} )  \end{pmatrix}\!\ket{\rho_0\!\otimes\! \tilde{\rho}_0}
% \end{align}
% For $k=n$ we have $P_k = P_{2n-k}$ (and also $T_k = T_{2n-k}$), moreover $P_{2n-k}$
% and the resulting "matrices" are one-dimensional: $M_n = \frac{1}{|P_{n}|}\tr(P_n\Lambda)$, $A_n = \frac{1}{|P_{k}|}\bra{E_x\otimes \Lambda\ct(\tilde{E}_x)}\ket{\mathrm{v}(P_n)}\bra{ \mathrm{v}(P_n)}\ket{\rho_0\otimes \tilde{\rho}_0}$.
Defining
\begin{align}
\lambda_k &= \frac{1}{|P_k|}\tr(P_k\Lambda),\\
A_k &= \frac{1}{|P_{k}|}\langle E_x\!\otimes\! \Lambda\ct(\tilde{E}_x) \ket{\mathrm{v}(P_{k})} \bra{\mathrm{v}(P_{k})} \rho_0\otimes \tilde{\rho}_0 \rangle,
\end{align}
 the expression for $f_k(m)$ becomes
\begin{equation}
f_k(m) = N_k^{-1}\sum_{x\in \{0,1\}^n} A_k \lambda^m_k,
\end{equation}
with $A_k = N^{-1}_k\sum_x A^x_k$ which is the expression we want.\\

Let's now consider the relation between the parameters $\lambda_k$ and the average fidelity. For this is it is useful to remember that the average fidelity of a trace preserving quantum channel $\Lambda$ can be written as
\begin{equation}
F(\Lambda) = \frac{2^{-n}\tr(\Lambda)+1}{2^{n}+1}.
\end{equation}
We can explicitly compute this trace from knowledge of the parameters $\lambda_k$, by noting that the projectors $P_k$ form a resolution of the identity (for the space of superoperators), and hence by construction
\begin{equation}
\sum_{k=0}^{2n} \frac{|P_k|}{2^{2n}}\lambda_k = \frac{1}{2^{2n}}\sum_{k=0}^{2n} \tr(P_k \Lambda) = \frac{1}{2^{2n}} \tr(\Lambda).
\end{equation}
Inserting this into the above and working out we obtain \cref{eq:fid} in the main text. Finally we want to consider the parameter $A_k$. In order to fit quantities of the form $A_k\lambda_k$ it is important that $A_k$ is non-zero. To ensure this we evaluate $A_k$ in the noise free limit, setting $\{E_x\}$ the be the computational basis and $\rho_0$ to be the all-zero state. The parameter $A_k$ is composed of quantities of the form $\tr(\dens{x} P_k(\dens{x}))$. We begin by noting that
\begin{equation}
\tr(\dens{x}P_k(\dens{x})) = 2^{-2n}\sum_{\substack{S\subset [2n]\\|S|=k}} |\tr(\dens{x}\gamma[S])|^2 = 2^{-2n}\sum_{\substack{S\subset [2n]\\|S|=k}}  |\tr(\dens{0}\gamma[S])|^2,
\end{equation}
since $\bra{x} = \bra{0}X_x$ for a bit-flip Pauli operator $X_x$ and $X_x\gamma[S]X_x\ct = \pm \gamma[S]$. Note that $|\tr(\dens{0}\gamma[S])|$ is one when $\gamma[S]$ is an all-$Z$ Pauli operator and zero otherwise. From the definition of the Majorana operators one can see that $|\tr(\dens{0}\gamma[S])|$ is always zero if $k$ is odd and that for even $k$ there are $\binom{n}{ k/2}$ choices for $S$ such that $\gamma[S]$ is an all-$Z$ Pauli.
With all the above we can evaluate the $A_k$ parameters as
\begin{align}
A_k &= N_{k}^{-1} \sum_{x\in\{0,1\}^n} \tr(\dens{x}P_k(\dens{x}))\tr(\dens{0}P_k(\dens{0}))\\
&= N_k^{-1} 2^{-n}\frac{ \binom{n}{ k/2}^2}{\binom{2n}{ k}},
\end{align}
for $k$ even.
Hence if we want $A_k=1$ in the noise free limit the normalization in the even $k$, $Z$-basis SPAM case must be set to
\begin{equation}
N_k^{(Z)}=  2^{-n}\frac{ \binom{n}{ k/2}^2}{\binom{2n}{ k}},
\end{equation}
where we also used that $|P_k| = \binom{2n}{k}$. We can repeat this exercise for odd $k$ and $X$-basis SPAM. Here the key expression is
\begin{align}
\tr\big(H\tn{n}\dens{x}H\tn{n}P_k(H\tn{n}\dens{x}H\tn{n})\big) &= 2^{-n}\sum_{\substack{S\subset [2n]\\|S|=k}} |\tr\big(H\tn{n}\dens{x}H\tn{n}\gamma[S]\big)|^2\\
 &= 2^{-n} \frac{ \binom{2n}{(k-1)/2 }^2}{\binom{2n}{ k}},
\end{align}
for all $x \in \{0,1\}^n$. The factor $\binom{2n}{(k-1)/2 }$ is obtained by considering the number of sets $S$ for which $|\tr\big(H\tn{n}\dens{x}H\tn{n}\gamma[S]\big)|^2=1$. This can only happen if $\gamma[S]$ contains only $X$ and $I$ tensor factors. We know that $\gamma_1 =X_1$ and also that $\gamma_{2i}\gamma_{2i+1} =X_iX_{i+1}$ (up to a phase). Since we require $k$ to be odd the set $S$ must consist of $1$, and $(k-1/2)$ tuples $2i, 2i+1$ with $i\in [1:n]$. Hence there are $\binom{n}{(k-1)/2}$ possible choices.
Using this we can compute
\begin{align}
A_k &= N_{k}^{-1} \sum_{x\in\{0,1\}^n} \tr(H\tn{n}\dens{x}H\tn{n}P_k(H\tn{n}\dens{x}H\tn{n}))\tr(\dens{+}P_k(\dens{+}))\\
&= N_k^{-1} 2^{-n}\frac{ \binom{n}{ \lfloor k/2\rfloor}^2}{\binom{2n}{ k}}.
\end{align}
Again imposing the unit condition gives the required normalization.\\

\section{Gate-dependent noise}\label{appsec:gate_dep}
In this section we discuss what happens when the gate-independent noise assumption made above breaks down. While a detailed calculation is out of the scope of this paper, we aim to show here that the conclusions reached for randomized benchmarking with arbitrary finite groups~\cite{framework} essentially carry over to the case of the matchgate group (which is not finite, but is compact). Consider a general map $\phi: \mc{M}_n\to \mc{S}_d$ from the orthogonal matchgates to the space of superoperators which assigns to each orthogonal matchgate a quantum channel. This is a general model of gate-dependent noise. Note  that this is not the most general possible model as we are ignoring non-Markovian and time-dependent effects. Within this model we can write the output of the matchgate benchmarking protocol as
\begin{align}
f_k(m) &= \int_{\OR(2n)} dQ_1\cdots Q_m\sum_{x\in \{0,1\}^{n}}\alpha_k(x,Q_1\cdots Q_m) \bra{\tilde{E}_x}\phi(Q_m) \cdots \phi(Q_1) \ket{\tilde{\rho}_0}\\
&= N_k^{-1}\int_{\OR(2n)} dQ_1\cdots Q_m\sum_{x\in \{0,1\}^{n}} \bra{E_x\otimes \tilde{E}_x}(P_k\otimes \id)\omega(Q_m)\otimes \phi(Q_m)) \cdots \omega(Q_1)\otimes \phi(Q_1) \ket{\rho_0\otimes \tilde{\rho}_0},
\end{align}
with $\tilde{E}_x,\tilde{\rho}_0$ noisy versions of their ideal counterparts. Noting that $P_k $ commutes with $\omega(Q)$ and that $P_k \omega(Q) = \omega_k(Q)$ where $\omega_k$ is the relevant irreducible subrepresentation, we have
\begin{equation}
 f_k(m) = \sum_{x\in \{0,1\}^{n}} \bra{E_x\otimes \tilde{E}_x}  \left(\int_{\OR(2n)}dQ\, \omega_k(Q)\otimes\phi(Q)\right)^m\ket{\rho_0\otimes \tilde{\rho}_0}.
 \end{equation}
We can now consider the operator $\int_{\OR(2n)}dQ\, \omega_k(Q)\otimes\phi(Q)$ as a perturbation of a rank one projector. In that case $f_k(m)$ will be well-described by the $m$-fold power of the largest eigenvalue of $\int_{\OR(2n)}dQ\, \omega_k(Q)\otimes\phi(Q)$. A candidate rank-one projector is given by the ideal implementation $\int_{\OR(2n)}dQ \omega_k(Q)\otimes\phi(Q)$. By an application of Schur's lemma it can be seen that this operator is equal to $\ket{v(P_k)}\bra{v(P_k)}$, a fact we have used already in the gate-independent derivation. More concretely, assume for a matrix norm $\|\cdot\|$  on the space $L(\Gamma_k)\otimes \mc{S}_d$ that
\begin{equation}
\left\|\int_{\OR(2n)}dQ\, \omega_k(Q)\otimes\phi(Q) - \int_{\OR(2n)}dQ\, \omega_k(Q)\otimes \omega(Q)\right\| \leq \delta,
\end{equation}
for some $\delta>0$. From the perturbation theory of $m$-fold matrix powers it can then be concluded that (provided $\delta$ is small enough)
 \begin{equation}\label{eq:delta}
f_k(m) = A_k \lambda_k^m + O(\delta^m),
\end{equation}
where $\lambda_k$ is the largest eigenvalue of the operator $\int_{\OR(2n)}dQ\,\omega_k(Q)\otimes\phi(Q)$.
In practice this means that even moderate deviations from the gate-independence assumption get suppressed exponentially quickly in the sequence length $m$. Hence even when the gate-independent noise assumption is relaxed the data obtained from a matchgate benchmarking experiment will be well described by a single exponential decay. \\

From the argument above it is not clear how small $\delta$ must be chosen, and what a physically reasonable choice of submultiplicative norm is. In  \cite{kong2021framework} (following \cite{framework}) it was shown (as a straightforward consequence of their theorem 1) that \cref{eq:delta} holds for standard randomized benchmarking with a compact group $\mathbb{G}$ provided
\begin{equation}
 \int dg \;\|{\phi(g) - \omega(g)}\|_\diamond \leq \delta\leq 1/9,
\end{equation}
holds (where $\|\cdot\|_{\diamond}$ is the diamond norm) and the integral is taken over the Haar measure. It is important to note here that the factor $1/9$ is likely suboptimal.

We can get a crude, rule-of-thumb indication for the size of $\delta$ by considering the decomposition of general matchgates into two-qubit matchgates given in \cref{eq:hurwitz}. Assuming that the single qubit $Z$ rotations are noiseless and that the two qubit rotations have an average diamond error of $\Delta$ we see by the triangle inequality and the sub-multiplicativity of the diamond norm that
\begin{equation}
 \int dQ \;\|{\phi(Q) - \omega(Q)}\|_\diamond  \leq n(n-1)\Delta.
\end{equation}
We can further estimate this by assuming that $\Delta \approx 1-F_{\rm avg}$ where $F_{\rm avg}$ is the average fidelity of the average two qubit gate. This is of course not generally true, and more or less corresponds to a "decoherent noise" assumption. In \cite{abrams2019implementation} a median two qubit gate fidelity of $\approx 97\%$ was reported, which we can slot in to give $\delta \approx n(n-1) 0.03$. Hence $\delta \leq 1/9$ for two qubits, and $\delta \approx 0.6$ for five qubits. We emphasize that this is a very crude order of magnitude estimation (on top of a suboptimal perturbation bound) meant to justify that $\delta$ can be small compared to $\lambda_k$ in reasonable circumstances, and should not be seen as an upper bound on the tolerance to gate-dependent noise of our protocol (especially in the context of larger $n$, where the triangle inequality used above becomes quite wasteful).

%=============================================================================
\section{Computation of correlation functions \texorpdfstring{$\alpha_k$}{alpha\_k}}\label{appsec:corr_func}
%=============================================================================
We show explicitly how to compute relevant quantities efficiently (in $n$). In particular the correlation functions given in \cref{eq:corr_func}. For this we will use some computational techniques from free (or Gaussian) fermionic states and operations (with which matchgates coincide). For an overview of these techniques see~\cite{bravyi2005lagrangian,bravyi2017complexity}. We begin by reviewing some identities. For a Majorana operator $\gamma[S]$ with $S\in [2n], |S| = k$ and a generalized matchgate $U(Q)$, with $Q\in \OR(2n)$ we have
\begin{equation}
U(Q)\gamma[S] U(Q)\ct = \sum_{S'\subset[2n], |S'|=k} \det\big(Q[S,S']\big) \gamma[S'],
\end{equation}
where $Q[S,S']$ denotes the matrix $Q$ with only the row indices in $S$ and column indices in $S'$ retained. Moreover we have for a computational basis state $\ket{x}$ and Majorana $\gamma[S]$ that
\begin{equation}\label{eq:wick}
\bra{x}\gamma[S]\ket{x} = \Pf(i M_x[S]),
\end{equation}
with $M[S]:= M[S,S]$, and where $\Pf$ denotes the Pfaffian. The matrix $M_x$ is defined by
\begin{equation}
M_x = \bigoplus_{i=1}^n \begin{pmatrix}0 & (-1)^x_i \\ (-1)^{x_i+1} & 0\end{pmatrix}.
\end{equation}
\Cref{eq:wick} is essentially Wick's theorem, and it extends to more general states.
The Pfaffian has the following three useful properties:
\begin{equation}
\Pf(A)^2 = \sqrt{\det(A)},
\end{equation}
for any even dimensional anti-symmetric matrix $A$,
\begin{equation}
\Pf\begin{pmatrix} M & C\\ -C^T & N
\end{pmatrix} = \Pf(M) \Pf(N  + C M^{-1} C^T),
\end{equation}
for invertible $M$, and
\begin{equation}
\Pf(QMQ^T) = \det(Q)\Pf(M).
\end{equation}
Moreover we will use a more advanced summation identity for Pfaffians, proven in~\cite{ishikawa2006applications}. Given anti-symmetric $2n\times 2n$ matrices $A,B$ and a $2n\times 2n $ matrix $C$ we have the polynomial identity
\begin{equation}\label{eq:gen_func_even}
\Pf(A)\Pf\big(\frac{A^{\mathrm{co}}}{\Pf(A)}+z^2(CBC^T)\big) = \sum_{s=0, s \text{ even}}^{n}z^{s}\sum_{\substack{S,S'\subset [2n]\\|S|=|S'|=s}}\Pf(A[S])\Pf(B[S'])\det(C[S,S']),
\end{equation}
where ${A}^{\mathrm{co}}$ denotes the co-Pfaffian matrix of $A$, which for invertible $A$ is given as ${A}^{\mathrm{co}} = \Pf(A)A^{-T}$. Note that this implies that ${R}^{\mathrm{co}} = \Pf(R) R$ if $R$ is orthogonal.
% We also have for an invertible $2n+2\times 2n+2$ matrix $A$ (indices going from $0$ to $2n+1$), a $2n+1\times 2n+1$ matrix $B$ (indices starting at $0$) and a  $2n\times 2n$ matrix $C$ (indices starting at one), that
% \begin{align}\label{eq:gen_func_odd}
% \Pf(A)\Pf(\frac{A^{co}}{\Pf(A)}+z^2(C')^T) &=
% \sum_{s=0, s \text{ even}}^{n}z^{s}\sum_{\substack{S,S'\subset [2n]\\|S|=|S'|=s}}\Pf(A[S])\Pf(B[S'])\det(C[S,S'])
% \\&\hspace{5em} +
% \sum_{s=0, s \text{ odd}}^{n}z^{s}\sum_{\substack{S,S'\subset [2n]\\|S|=|S'|=s}}\Pf(A[\{0\}\cup S])\Pf(B[\{0\}\cup S'])\det(C[S,S'])
% \end{align}
% where the matrix $C'$ is defined as
% \begin{equation}
% C' = \begin{pmatrix} 0 & z C B[\{0\},[2n]]^T & 0 \\ z B[[2n],\{0\}]^TC  & z^2 CBC^T & 0 \\ 0 & 0 & 0 \end{pmatrix}
% \end{equation}
% where $B[\{0\},[2n]]$ is the zeroth row of $B$ with the zeroth entry removed.
With these identities in hand we move on to compute $\alpha_k(x,Q)$ for even $k$ and $Z$-basis SPAM. From the definition we have
\begin{align*}
  \alpha_k(x, Q) = N_k^{-1} \sum_{\vsp\mathclap{S \subseteq [2n], \lvert S\rvert = k}} \beta_S(x, I) \beta_S(0, Q),
\end{align*}
with $\beta_S(x,Q) = 2^{-n/2} \tr\bigl( \gamma[S] U(Q) E_x U(Q)\ct \bigr)$. Using $E_x = \dens{x}$ we see that
\begin{align*}
  \alpha_k(x, Q) = \frac{2^{-n}}{N_k}\;\;\;\sum_{\vsp\mathclap{S \subseteq [2n], \lvert S\rvert = k}}  \Pf(i M_x[S]) \sum_{\vsp\mathclap{S' \subseteq [2n], \lvert S'\rvert = k}} \Pf(i M_0[S'])\det(Q[S',S]).
\end{align*}
To this we can apply the summation formula \cref{eq:gen_func_even} to conclude that
\begin{equation}
\alpha_k(x, Q)  = -\frac{2^{-n}}{N_k}\frac{1}{k!}\, \partial^k_z \left(\Pf(i M_x)\Pf(i M_x+z^2i QM_0Q^T)\right)\bigg|_{z=0},
\end{equation}
which can easily be evaluated numerically.
Next we calculate the correlator for odd $k$ and $X$-basis SPAM, which is more complicated.
We begin by calculating the quantity
\[
\beta_x^Q[S] = \tr\big(H\tn{n}\dens{x}H\tn{n}U(Q)\gamma[S]U(Q)\ct\big),
\]
for $x\in \{0,1\}^{n}$ and arbitrary $U(Q)\in \mc{M}_n^+$.
We first note that for every $X$ basis state $H\tn{n}\ket{x}$ there exists~$Q_x \in \OR(2n)$ s.t.\ $U(Q_x)H\tn{n}\ket{x} = \frac{I+ i\gamma_{1}}{\sqrt{2}}\ket{0}$. Using this and the determinant expression for the action of $U(Q)$ we get
\[
\beta_x^Q[S] = 2^{-1}\;\;\;\sum_{\mathclap{S' \subseteq [2n], \lvert S'\rvert = k}} \det\big(Q_xQ[S,S']\big)\ket{0}(I + i\gamma_1)\gamma[S'](I+ i\gamma_1)\ket{0} = \sum_{\vsp\mathclap{S' \subseteq [2n], \lvert S'\rvert = k}} \det\big(Q_xQ[S,S']\big)\mc{I}(1\in S')\Pf(i M_0[S'/\{1\}]).
\]
This gives for the correlation function
\begin{equation}
\alpha_k(x, Q)  = \frac{2^{-n}}{N_k}\sum_{\substack{\vsp\mathclap{S \subseteq [2n], \lvert S\rvert = k}\\\mathclap{S' \subseteq [2n], \lvert S'\rvert = k}\\\mathclap{S'' \subseteq [2n], \lvert S''\rvert = k}}}\det\big(Q_x[S,S']\big)\det\big(Q_0Q[S,S']\big)\mc{I}(1\in S')\mc{I}(1\in S'') \Pf(iM_0[S'/\{1\}])\Pf(iM_0[S''/\{1\}]),
\end{equation}
which we can simplify using the Cauchy-Binet identity to
\begin{equation}
\alpha_k(x, Q)  = \frac{2^{-n}}{N_k}\sum_{\substack{\vsp\mathclap{S' \subseteq [2n], \lvert S'\rvert = k}\\\mathclap{S'' \subseteq [2n], \lvert S''\rvert = k}}}\det\big(Q_x^T Q_0 Q[S'',S']\big)\mc{I}(1\in S')\mc{I}(1\in S'') \Pf(iM_0[S'/\{1\}])\Pf(iM_0[S''/\{1\}]).
\end{equation}
% This is close to the entries of the generating function in \cref{eq:gen_func_odd}, but not entirely.
Now note that $\mc{I}(1\in S') \Pf(iM_0[S'/\{1\}])$ is always zero if $2\in S'$ (this follows directly from the definition of $M_0$ and the fact that the Pfaffian is always zero for non-full rank matrices).
% We can remedy this by introducing the matrices
% \begin{equation}
% M_0' = \begin{pmatrix} 0 & e_1  \\-e_1^T & M_0\end{pmatrix} , \;\;\;\;\;\; M_0'' = \begin{pmatrix} M_0' & e_0'^T\\-e_{0} & 0 \end{pmatrix}
% \end{equation}
% where $e_1 = (1, 0, \ldots ,0)$ of length $2n$ and $e_0' = (1, 0, \ldots ,0)$ of length $2n+1$. Both of these matrices are zero-indexed. One can see that for sets $S'\in [2n]$
% \begin{equation}
% \mc{I}(1\in S') \Pf(M_0[S'/\{1\}]) = \Pf(M_0'[S'\cup\{1\}]) = \Pf(M_0''[S'\cup\{1\}]).
% \end{equation}
Hence we can rewrite the correlation function as
\begin{equation}
\alpha_k(x, Q)  = \frac{2^{-n}}{N_k}\sum_{\substack{\vsp\mathclap{S' \subseteq [3:2n], \lvert S'\rvert = k-1}\\\mathclap{S'' \subseteq [3:2n], \lvert S''\rvert = k-1}}}\det\big(Q_x^T Q_0 Q[S'',S']\big)\Pf(iM_0[S'])\Pf(iM_0[S'']).
\end{equation}
Defining the matrices $\widetilde{Q_x^T Q_0 Q} =(Q_x^T Q_0 Q)\big[[3\!\!:\!\!2n]\big] $ and  $\widetilde{M_0} = M_0\big[[3\!\!:\!\!2n]\big]$ we can apply \cref{eq:gen_func_even} and obtain the correlation function as a $k-1$'th derivative of a Pfaffian generating function involving matrices of dimension $2(n-1)$.
\begin{equation}
\alpha_k(x,Q) = \frac{2^{-n}}{N_k}\frac{1}{(k-1)!} \partial_z^{k-1} \Pf\big(i\widetilde{M_0} + z^2 i\widetilde{Q_x^TQ_0 Q}\widetilde{M_0}(\widetilde{Q_x^TQ_0 Q})^T\big)\bigg|_{z=0},
\end{equation}
which allows for direct numerical calculation.

\end{document}